\pdfoutput=1
\documentclass{article}
\usepackage{amsmath}
\usepackage{amsthm}
\usepackage{mathptmx}
\usepackage[libertine]{newtxmath}
\usepackage[T1]{fontenc}
\usepackage{geometry}
\usepackage{color}
\usepackage{array}
\usepackage{verbatim}
\usepackage{pifont}
\usepackage{float}
\usepackage{mathtools}
\usepackage{esint}
\PassOptionsToPackage{normalem}{ulem}
\usepackage{ulem}
\usepackage{nameref}

\makeatletter

\providecommand{\tabularnewline}{\\}

\usepackage{etoolbox}
\usepackage[bottom,hang]{footmisc}
\usepackage{xcolor}
\usepackage{hyperref}	
\hypersetup
{
  unicode=true,
  colorlinks=true,
  linkcolor=[rgb]{0 0 .8},
  citecolor=[rgb]{0 .8 0},
  urlcolor=[rgb]{.8 0 .8},
  breaklinks=true,  
  pdfstartview=,
  bookmarksopen=true,
}
\usepackage[numbered]{bookmark}
\usepackage[all]{hypcap}
\usepackage[nameinlink]{cleveref}	
\AtBeginDocument{\let\ref\cref}
\crefname{section}{Section}{Sections}
\crefname{figure}{Figure}{Figures}
\crefname{example}{Example}{Examples}
\makeatletter
\patchcmd\@combinedblfloats
  {\box\@outputbox}
  {\unvbox\@outputbox}
  {}{\errmessage{\noexpand\@combinedblfloats could not be patched}}
\makeatother
\usepackage{graphicx}
\let\oldsqrt\sqrt
\renewcommand{\sqrt}[2][\ \,\,]{{\!\!\oldsqrt[\raisebox{.1em}{\scalebox{.7}{$#1$}}]{#2}\,}}
  \let\originalleft\left
  \let\originalright\right
  \renewcommand{\left}{\mathopen{}\mathclose\bgroup\originalleft}
  \renewcommand{\right}{\aftergroup\egroup\originalright}
\binoppenalty=\maxdimen
\DeclareMathSizes{5}{5}{4}{3}
\DeclareMathSizes{6}{6}{5}{4}
\DeclareMathSizes{7}{7}{5.5}{4}
\DeclareMathSizes{8}{8}{6.5}{5}
\DeclareMathSizes{9}{9}{7.5}{6}
\DeclareMathSizes{10}{10}{8.5}{7}
\DeclareMathSizes{10.95}{10.95}{9}{7}	
\DeclareMathSizes{12}{12}{10}{8}
\DeclareMathSizes{14.4}{14.4}{12}{10}
\DeclareMathSizes{17.28}{17.28}{14.5}{11.5}
\DeclareMathSizes{20.74}{20.74}{17}{13.5}
\DeclareMathSizes{24.88}{24.88}{20}{16}
\makeatletter
\AtBeginDocument
{
  \check@mathfonts
  \fontdimen13\textfont2 = \fontdimen13\textfont2		
  \fontdimen14\textfont2 = \fontdimen13\textfont2		
  \fontdimen15\textfont2 = \fontdimen13\textfont2		
  \fontdimen16\textfont2 = 1.3\fontdimen16\textfont2	
  \fontdimen17\textfont2 = \fontdimen16\textfont2		
  \fontdimen14\scriptfont2 = \fontdimen13\scriptfont2	
  \fontdimen15\scriptfont2 = \fontdimen13\scriptfont2
  \fontdimen16\scriptfont2 = 1.3\fontdimen16\scriptfont2
  \fontdimen17\scriptfont2 = \fontdimen16\scriptfont2
  \fontdimen14\scriptscriptfont2 = \fontdimen13\scriptscriptfont2	
  \fontdimen15\scriptscriptfont2 = \fontdimen13\scriptscriptfont2
  \fontdimen16\scriptscriptfont2 = 1.3\fontdimen16\scriptscriptfont2
  \fontdimen17\scriptscriptfont2 = \fontdimen16\scriptscriptfont2
}
\makeatother
\usepackage[medium,compact]{titlesec}
\usepackage[skip=.5em]{caption}
\newcommand{\setmulength}[2]{#1=#2\relax}	
\setmulength{\thinmuskip}{2mu plus 1mu minus 1mu}
\setmulength{\medmuskip}{2mu plus 1mu minus 1mu}
\setmulength{\thickmuskip}{4mu plus 1mu minus 2mu}
\AtBeginEnvironment{pmatrix}
{
  \setlength{\arraycolsep}{.4em}
  
}
\let\oldsection\section
\let\oldsubsection\subsection
\let\oldsubsubsection\subsubsection
\def\section{\addpenalty{-4000}\oldsection}	
\def\subsection{\addpenalty{-3000}\oldsubsection}
\def\subsubsection{\addpenalty{-2000}\oldsubsubsection}
\let\oldpar\par
\def\par{\oldpar\addpenalty{-1000}}
  \usepackage{enumitem}
  \setlistdepth{14}
  \newlength{\listspace}
  \setlist{topsep=\listspace,itemsep=\listspace,parsep=0em,partopsep=0em}
  \setlist[enumerate]{leftmargin=2em}
  \setlist[itemize]{leftmargin=1.5em}
  \AtBeginEnvironment{itemize}{\interlinepenalty=3000}
  \AtBeginEnvironment{enumerate}{\interlinepenalty=3000}
\makeatletter
\newcommand{\justified}{%
  \rightskip\z@skip%
  \leftskip\z@skip}
\makeatother
\date{}
\usepackage{upgreek}
\usepackage{calc}	
\newlength{\linespace}

\newlength{\parspace}

\newtheoremstyle{lwq}
  {0em}	
  {0em}	
  {\normalfont}	
  {0em}	
  {\bfseries}	
  {.}	
  {.3em plus .2em}	
  {\thmname{#1}\thmnumber{ #2}\thmnote{ (#3)}}	
\newtheoremstyle{lwqprf}
  {0em}	
  {0em}	
  {\normalfont}	
  {0em}	
  {\itshape}	
  {.}	
  {.3em plus .2em}	
  {\thmname{#1}\thmnote{ (#3)}}	
\newlength{\thmspace}
\newlength{\prfspace}
\newcommand\thmbegin{\par\addvspace{\thmspace}\addvspace{\parskip}\addpenalty{-1500}}
\newcommand\prfbegin{\par\addvspace{\prfspace}\addvspace{\parskip}\addpenalty{-500}}
\newcommand\thmend{\par\addvspace{\thmspace}\addpenalty{-500}}	
\newcommand\prfend{\par\addvspace{\prfspace}\addpenalty{-500}}	
\newcommand{\theoremname}{Theorem}
\theoremstyle{lwq}\newtheorem{thm}{\protect\theoremname}
\newcounter{subtheorem}[thm]
\AtBeginEnvironment{thm}{\thmbegin} \AtEndEnvironment{thm}{\thmend}
\newcommand{\definitionname}{Definition}
\theoremstyle{lwq}\newtheorem{defn}[thm]{\protect\definitionname}
\AtBeginEnvironment{defn}{\thmbegin} \AtEndEnvironment{defn}{\thmend}
\newcommand{\lemmaname}{Lemma}
\theoremstyle{lwq}\newtheorem{lem}[thm]{\protect\lemmaname}
\AtBeginEnvironment{lem}{\thmbegin} \AtEndEnvironment{lem}{\thmend}
\renewcommand{\proofname}{Proof}
\theoremstyle{lwqprf}\newtheorem{prf}{\protect\proofname}
\renewenvironment{proof}[1][]{\prfbegin\begin{prf}[#1]\pushQED{\qed}}{\popQED\end{prf}\prfend}
\renewcommand{\qed}{\hfill{}\hspace{3em minus 1em}\qedsymbol}
\renewcommand{\qed}{}	
\newcommand{\remarkname}{Remark}
\theoremstyle{lwq}
\AtBeginEnvironment{rem}{\thmbegin} \AtEndEnvironment{rem}{\thmend}
\theoremstyle{lwq}\newtheorem*{rem*}{\protect\remarkname}
\AtBeginEnvironment{rem*}{\thmbegin} \AtEndEnvironment{rem*}{\thmend}
\newenvironment{centerbox}
{\begin{centering}}
{\par\end{centering}}
\newcommand{\corollaryname}{Corollary}
\theoremstyle{lwq}\newtheorem{cor}[thm]{\protect\corollaryname}
\AtBeginEnvironment{cor}{\thmbegin} \AtEndEnvironment{cor}{\thmend}
\newcommand{\setblockspace}
{
  \setlength{\topsep}{0em}
  \setlength{\itemsep}{\listspace}
  \setlength{\parsep}{0em}
  \setlength{\partopsep}{\listspace}
}
\newcommand{\setblockflush}{}
\newenvironment{block}
  {\list{}{\leftmargin1.5em\setblockspace\interlinepenalty3000}\setblockflush}
  {\endlist}

\AtBeginDocument{\fontsize{10}{13}\selectfont}
\crefname{footnote}{}{}
\crefformat{footnote}{#2\footnotemark[#1]#3}
\AtBeginDocument{

}

\makeatother

\begin{document}
\renewcommand{\qed}{\hfill{}\hspace{2em minus 1em}\scalebox{1.4}{$\diamond$}}

\newcommand\br{\addpenalty{-1000}}

\newcommand\step{\crefname{enumi}{step}{}}

\global\long\def\nn{\mathbb{N}}
\global\long\def\zz{\mathbb{Z}}
\global\long\def\qq{\mathbb{Q}}
\global\long\def\rr{\mathbb{R}}

\global\long\def\wi{\subseteq}
\global\long\def\co{\supseteq}
\global\long\def\nwi{\nsubseteq}
\global\long\def\nco{\nsupseteq}
\global\long\def\none{\varnothing}
\global\long\def\less{\smallsetminus}

\global\long\def\ii{\mathbf{1}}
\global\long\def\pp{\mathbf{P}}
\global\long\def\ee{\mathbf{E}}
\global\long\def\vv{\mathbf{Var}}
\global\long\def\cv{\mathbf{Cov}}

\global\long\def\floor#1{\left\lfloor #1\right\rfloor }
\global\long\def\ceil#1{\left\lceil #1\right\rceil }

\global\long\def\cond#1#2#3{\left(\vphantom{#1#2#3}\right.\,#1\mathrel{\,?\,}\allowbreak#2\mathrel{\,:\,}\allowbreak#3\,\left.\vphantom{#1#2#3}\right)}

\global\long\def\f#1{\operatorname{#1}}

\global\long\def\a#1{.\text{#1}}

\global\long\def\e{\upvarepsilon}

\global\long\def\tbox#1{\boxed{\text{#1}}}

\global\long\def\stbox#1{\scalebox{0.8}{\boxed{\text{#1}}}}

\global\long\def\smbox#1{\scalebox{0.8}{\boxed{#1}}}

\global\long\def\grey#1{\textcolor{gray}{#1}}

\global\long\def\vsp#1{\vspace{#1em}}

\noindent \begin{center}
\textbf{\huge{}Parallel Working-Set Search Structures}
\par\end{center}{\huge \par}

\noindent \begin{center}
\begin{tabular}{>{\centering}p{0.3\textwidth}>{\centering}p{0.3\textwidth}>{\centering}p{0.3\textwidth}}
\textbf{\large{}Kunal Agrawal}{\large \par}

Washington University in St. Louis & \textbf{\large{}Seth Gilbert}{\large \par}

National University of Singapore & \textbf{\large{}Wei Quan Lim}{\large \par}

National University of Singapore\tabularnewline
\end{tabular}
\par\end{center}

\section*{Keywords}

Data structures, parallel programs, dictionaries, comparison-based
search, distribution-sensitive algorithms

\section*{Abstract}

In this paper~\footnote{This is the authors' version of a paper submitted to the 30th ACM
Symposium on Parallelism in Algorithms and Architectures (SPAA '18).
It is posted here by permission of ACM for your personal or classroom
use. Not for redistribution. The definitive version can be found at
https://doi.org/10.1145/3210377.3210390. \copyright~ 2018 Copyright
is held by the owner/author(s). Publication rights licensed to ACM.} we present two versions of a parallel working-set map on $p$ processors
that supports searches, insertions and deletions. In both versions,
the total work of all operations when the map has size at least $p$
is bounded by the \emph{working-set bound}, i.e., the cost of an item
depends on how recently it was accessed (for some linearization):
accessing an item in the map with recency $r$ takes $O(1+\log r)$
work. In the simpler version each map operation has $O\left((\log p)^{2}+\log n\right)$
span (where $n$ is the maximum size of the map). In the pipelined
version each map operation on an item with recency $r$ has $O\left((\log p)^{2}+\log r\right)$
span. (Operations in parallel may have overlapping span; span is additive
only for operations in sequence.)

Both data structures are designed to be used by a dynamic multithreading
parallel program that at each step executes a unit-time instruction
or makes a data structure call. To achieve the stated bounds, the
pipelined version requires a weak-priority scheduler, which supports
a limited form of 2-level prioritization. At the end we explain how
the results translate to practical implementations using work-stealing
schedulers.

To the best of our knowledge, this is the first parallel implementation
of a \emph{self-adjusting} search structure where the cost of an operation
adapts to the access sequence. A corollary of the working-set bound
is that it achieves \emph{work static optimality}: the total work
is bounded by the access costs in an optimal static search tree.

\section{Introduction}

\textbf{Map} (or \textbf{dictionary}) data structures, such as binary
search trees and hash tables, support inserts, deletes and searches/updates
(collectively referred to as accesses) and are one of the most used
and studied data structures. In the comparison model, balanced binary
search trees such as AVL trees and red-black trees provide a performance
guarantee of $O(\log n)$ worst-case cost per access for a tree with
$n$ items. Other kinds of balanced binary trees provide probabilistic
or amortized performance guarantees, such as treaps and weight-balanced
trees.

\textbf{Self-adjusting maps}, such as splay trees~\cite{sleator1985splay},
adapt their internal structure to the sequence of operations to achieve
better performance bounds that depend on various properties of the
access pattern (see \cite{elmasry2013distsensitive} for a hierarchical
classification). Many of these data structures make it cheaper to
search for recently accessed items (temporal locality) or items near
to previously accessed items (spatial locality). For instance, the
working-set structure described by Iacono in \cite{iacono2001wstree}
has the working-set property (which captures temporal locality); it
takes $O(\log r+1)$ time per operation with access rank $r$ (\ref{def:access-rank}),
so its total cost satisfies the working-set bound (\ref{def:work-set-bound}).
\begin{defn}[Access Rank]
\label{def:access-rank} Define the \textbf{access rank} of an operation
in a sequence of operations on a map $M$ as follows. The access rank
for a successful search for item $x$ is the number of distinct items
in $M$ that have been searched for or inserted since the last prior
operation on $x$ (including $x$ itself). The access rank for an
insertion, deletion or unsuccessful search is always $n+1$, where
$n$ is the current size of $M$.
\end{defn}

\begin{defn}[Working-Set Bound]
\label{def:work-set-bound} Given any sequence $L$ of $N$ map operations,
we shall use $W_{L}$ to denote the \textbf{working-set bound} for
$L$, defined by $W_{L}=\sum_{i=1}^{N}\left(\log r_{i}+1\right)$
where $r_{i}$ is the access rank of the $i$-th operation in $L$
when $L$ is performed on an empty map.
\end{defn}

\subsection*{Parallel search structures}

Our goal in this paper is to design efficient self-adjusting search
structures that can be used by parallel programs. Even designing a
non-adjusting parallel or concurrent search structure is quite challenging,
and there has been a lot of research on the topic.

There are basically two approaches. In the concurrent computing world,
processes independently access the data structure using some variety
of concurrency control (e.g., locks) to prevent conflicts. In the
parallel computing world, data structures are designed to support
executing operations in parallel, either individually or in batches.

For example, in the concurrent computing world, Ellen et al.~\cite{EllenPR10}
show how to design a non-blocking binary search tree, with later work
generalizing this technique~\cite{BrownER14} and analyzing the amortized
complexity~\cite{EllenFH14}. However, these data structures do not
maintain balance in the tree (i.e., the height can get large) and
their cost depends on the number of concurrent operations.

An alternate approach (that bears some similarity to the implicit
batching that we use) is software combining~\cite{FatourouKa12,HendlerInSh10,OyamaTaYo99},
where each processor inserts a request in a shared queue and a single
processor sequentially executes all outstanding requests later. These
works provide empirical efficiency but not worst-case bounds.

Another notable example is the CBTree~\cite{AfekKK12,AfekKK14},
a concurrent splay tree that in real experiments achieves surprisingly
good performance --- leading to an interesting hypothesis that self-adjustment
may be even more valuable (in practice) in concurrent settings than
sequential settings. However, the CBTree does not guarantee that it
maintains the proper `frequency balance', and hence does not provide
the guarantees of a splay tree (despite much experimental success).

In the parallel computing world, there are several classic results
in the PRAM model. Paul et al.~\cite{paul1983paradict} devised a
parallel 2-3 tree such that $p$ synchronous processors can perform
a batch of $p$ operations on a parallel 2-3 tree of size $n$ in
$O(\log n+\log p)$ time. Blelloch et al.~\cite{BlellochRe97} show
how pipelining can be used to increase parallelism of tree operations.
Also, (batched) parallel priority queues~\cite{BrodalTrZa98,CrauserMeMe98,DriscollGaSh88,Sanders98}
have been utilized to give efficient parallel algorithms such as for
shortest-path and minimum spanning tree~\cite{BrodalTrZa98,DriscollGaSh88,PaigeKr85}.

More recently, in the dynamic multithreading model, there have been
several elegant papers on parallel treaps~\cite{BlellochR98} and
how to parallelize a variety of different binary search trees~\cite{BlellochFS16}
supporting unions and intersections, and also work on how to achieve
batch parallel search trees with optimal work and span~\cite{AkhremtsevS16}.
Other batch parallel search trees include red-black trees~\cite{FriasSi07}
and weight-balanced B-trees~\cite{EKS14}. (We are unaware of any
batched self-adjusting data structures.)

And yet, such concurrent/parallel map data structures can be difficult
to use; the programmer cannot simply treat it as a black box and use
atomic map operations on it from within an ordinary parallel program.
Instead, she must carefully coordinate access to the map.

\subsection*{Implicit batching}

Recently, Agrawal et al.~\cite{agrawal2014batcher} introduced the
idea of \textbf{\textit{\emph{implicit batching}}}. Here, the programmer
writes a parallel program that uses a black box data structure, treating
calls to the data structure as basic operations. In addition, she
provides a data structure that supports batched operations (e.g.,
search trees in~\cite{BlellochR98,BlellochFS16}). The runtime system
automatically stitches these two components together, ensuring efficient
running time by creating batches on the fly and scheduling them appropriately.
This idea of implicit batching provides an elegant solution to the
problem of parallel search trees.

\subsection*{Our goals}

Our goal is to extend the idea of implicit batching to self-adjusting
data structures --- and more generally, to explore the feasibility
of the implicit batching approach for a wider class of problems. In~\cite{agrawal2014batcher},
they show how to apply the idea to uniform-cost data structures (where
every operation has the same cost).~\footnote{They also provide some bounds for amortized data structures, where
queries do not modify the data structure.} In a self-adjusting structure, some operations are much cheaper than
others, and additionally every operation may modify the data structure
(unlike say AVL/red-black trees where searches have no effect on the
structure), which makes parallelizing it much harder.

We present in this paper, to the best of our knowledge, the first
parallel self-adjusting search structure that is distribution-sensitive
with worst-case guarantees. In particular, we design two versions
of a parallel map whose total work is essentially bounded by the \nameref{def:work-set-bound}
for some linearization $L$ of the operations (that respects the dependencies
between them).

\subsection*{Parallel Programming Model}

The parallel data structures in this paper can be used in the scenario
where a parallel program accesses data structures expressed through
\textbf{dynamic multithreading} (see~\cite[Ch. 27]{CormenLeRi09}),
which is the case in many parallel languages and libraries, such as
Cilk dialects~\cite{Cilk,IntelCilkPlus13}, Intel TBB~\cite{TBB},
Microsoft Task Parallel Library~\cite{TPL} and subsets of OpenMP~\cite{OpenMP}.
The programmer expresses algorithmic parallelism through parallel
programming primitives such as fork/join (also spawn/sync), parallel
loops and synchronized methods, and does not provide any mapping from
subcomputations to processors.

These types of programs are typically scheduled using a \textbf{greedy
scheduler}~\cite{Brent,Graham} or a nearly greedy scheduler such
as \textbf{work-stealing scheduler} (e.g.,~\cite{blumofe1999worksteal})
provided by the runtime system. A greedy scheduler guarantees that
at each time step if there are $k$ available tasks then $\min(k,p)$
of them are completed.

We analyze our two data structures in the context of a greedy scheduler
and a weak-priority scheduler (respectively). A \textbf{weak-priority
scheduler} has two priority levels, and at each step at least half
the processors greedily choose high-priority tasks and then low-priority
tasks --- if there are at most $\frac{1}{2}p$ high-priority tasks,
then all are executed. We discuss in \ref{sec:work-steal} how to
adapt these results for work-stealing schedulers.

\section{Main Results}

We present two parallel working-set maps that can be used with any
parallel program $P$, whose actual execution is captured by a \textbf{program
DAG} $D$ where each node is a unit-time instruction or a call to
some data structure $M$, called an \textbf{$M$-call}, that blocks
until the answer is returned, and each edge represents a dependency
due to the parallel programming primitives. Let $T_{1}$ be the total
number of nodes in $D$, and $T_{\infty}$ be the number of nodes
on the longest path in $D$.

Both designs take work nearly proportional to the \nameref{def:work-set-bound}
$W_{L}$ for some legal linearization $L$ of $D$, while having good
parallelism. (We assume that each key comparison takes $O(1)$ steps.)

The first design, called $M_{1}$, is a simpler batched data structure.
\begin{thm}[$M_{1}$ Performance]
\label{thm:M1-perf} If $P$ uses only $M_{1}$ (i.e., no other data
structures), then its running time on $p$ processes using any greedy
scheduler is 
\[
O\left(\frac{T_{1}+W_{L}+e_{L}\cdot\log p}{p}+T_{\infty}+d\cdot\left((\log p)^{2}+\log n\right)\right)
\]
(as $n,p\to\infty$) for some linearization $L$ of $D$, where $d$
is the maximum number of $M_{1}$-calls on any path in $D$, and $n$
is the maximum size of the map, and $e_{L}$ is the number of \textbf{small-ops},
defined as operations in $L$ that are performed on the map when its
size is less than $p$.
\end{thm}
Notice that if $M_{1}$ is replaced by an ideal concurrent working-set
map (one that does the same work as the sequential working-set map
if we ran the program according to linearization $L$), then running
$P$ on $p$ processors according to the linearization $L$ takes
$\Omega(T_{opt})$ worst-case time where $T_{opt}=\frac{T_{1}+W_{L}}{p}+T_{\infty}$.
Also, we very likely have $e_{L}\ll\frac{W_{L}}{\log p}$ in practice,
and so can usually ignore the $e_{L}\cdot\log p$ term. Thus $M_{1}$
gives an essentially optimal time bound except for the ``span term''
$d\cdot\left((\log p)^{2}+\log n\right)$, which adds $O\left((\log p)^{2}+\log n\right)$
time per $M_{1}$-call along some path in $D$. In short, the parallelism
of $M_{1}$ is within a factor of $O\left((\log p)^{2}+\log n\right)$
of the optimal.

The second design, called $M_{2}$, uses a more complex pipelined
data structure design as well as a weak-priority scheduler (see \ref{sub:weak-priority-sched})
to provide a better bound on the ``span term''.
\begin{thm}[$M_{2}$ Performance]
\label{thm:M2-perf} If $P$ uses only $M_{2}$, then its running
time on $p$ processes using any weak-priority scheduler is 
\[
O\left(\frac{T_{1}+W_{L}+e_{L}\cdot\log p}{p}+T_{\infty}+d\cdot(\log p)^{2}+s_{L}\right)
\]
for some linearization $L$ of $D$, where $d,e_{L}$ are defined
as in \ref{thm:M1-perf}, and $s_{L}$ is the weighted span of $D$
where each map operation is weighted by its cost according to $W_{L}$.
Specifically, each map operation in $L$ with access rank $r$ is
given the weight $\log r+1$, and $s_{L}$ is the maximum weight of
any path in $D$.
\end{thm}
Compared to $M_{1}$, the ``work term'' $\frac{T_{1}+W_{L}+e_{L}\cdot\log p}{p}$
is unchanged, but the ``span term'' for $M_{2}$ has no $\log n$
term. Since running $P$ on $p$ processors according to the linearization
$L$ takes $\Omega(T_{opt}+s_{L})$ worst-case time, $M_{2}$ gives
an essentially optimal time bound up to an extra $O\left((\log p)^{2}\right)$
time per map operation along some path in $D$, and hence $M_{2}$
has parallelism within an $O\left((\log p)^{2}\right)$ factor of
optimal.

\section{Central Ideas}

\label{sec:challenges}

We shall now sketch the intuitive motivations behind $M_{1}$ and
$M_{2}$.

It starts with \textbf{\textit{Iacono's sequential working-set structure}},
which contains a sequence of balanced binary search trees $t_{1},t_{2},...,t_{l}$
where tree $t_{i}$ for $i<l$ contains $2^{2^{i}}$ items and hence
has height $\Theta(2^{i})$. The invariant maintained is that the
$r$ most recently accessed items are in the first $\log\log r$ trees.
A search on a key proceeds by searching in each tree in the sequence
in order until the key is found in tree $t_{k}$. After a search,
the item is moved to $t_{1}$ and then for each $i<k$, the least
recently accessed item from tree $t_{i}$ is moved to tree $t_{i+1}$.
By the invariant, any item in the map with recency $r$ will take
$O(\log r+1)$ time to access. Each insertion or deletion can be easily
carried out in $O(\log(n+1)+1)$ time while preserving the invariant.

\textbf{\textit{The challenge is to `parallelize' this working-set
structure while preserving the total work.}} The first step is to
\textbf{\textit{process operations in batches}}, using a batched search
structure in place of each `tree'.

The problem is that, if there are $b$ searches for the same item
in the last tree, then according to the working-set bound these $b$
operations should take $O(\log n+b)$ work. But if these operations
all happen in parallel and end up in the same batch, and we execute
this batch naively, then each operation will go through the entire
structure leading to $\Omega(b\cdot\log n)$ work.

Therefore, in order to get the desired bound, we must \textbf{\textit{combine
duplicate accesses}} in each batch. But naively sorting a batch of
$b$ operations takes $\Theta(b\cdot\log b)$ work. To eliminate this
as well, $M_{1}$ (\ref{sec:PWM1}) uses a novel \textbf{\textit{entropy-sorting
algorithm}}, and a careful analysis yields the desired work bound.

Next, we cannot simply apply the generic ``implicit batching'' transformation
in \cite{agrawal2014batcher} to $M_{1}$, because the Batcher Bound
(Theorem~1 in \cite{agrawal2014batcher}) would give an expected
running time of $O\left(\frac{T_{1}+W+N\cdot s}{p}+T_{\infty}+d\cdot s\right)$
for $N$ map operations, where $W$ is the work done by $M_{1}$,
and $s$ is the worst-case span of a size-$p$ batch. The problem
is that $s$ is $\Omega(\log n)$, because a batch with a search for
an item in the last tree has span $\Omega(\log n)$.

Firstly, this means that the $N\cdot s$ term would be $\Omega(N\cdot\log n)$,
and so the Batcher Bound would be no better than for a batched binary
search tree. Secondly, the $d\cdot s$ term would be $\Omega(d\cdot\log n)$.
$M_{1}$ has the same span term, because if a cheap operation is `blocked'
by the previous batch that has an expensive operation, then the span
of the cheap operation could be $\Omega(\log n)$. To reduce this,
we improve $M_{1}$ to $M_{2}$ using an intricate \textbf{\textit{pipelining
scheme}} (explained in \ref{sec:PWM2}) so that a cheap operation
is `blocked' less by the previous batch.

\section{Parallel Computation Model}

\label{sec:model}

In this section, we describe how the parallel program $P$ generates
an execution DAG, how we measure the cost of a given execution DAG,
and issues related to the chosen memory model.

\subsection*{Execution DAG}

The actual complete execution of $P$ can be captured by the \textbf{execution
DAG} $E$ (which may be schedule-dependent), in which each node is
a unit-time instruction and the directed edges represent the underlying
computation flow (such as constrained by forking/joining of threads
and acquiring/releasing of locks).%
{} At any point during the execution of $P$, a node in the program/execution
DAG is said to be \textbf{ready} if its parent nodes have been executed.
An \textbf{active thread} is simply a ready node in $E$, while a
\textbf{suspended thread} is a terminal node in $E$.

The \textbf{program DAG} $D$ captures the high-level execution of
$P$, but interaction between data structure calls is only captured
by the execution DAG. We further assume that all the data structures
are (implicitly) batched data structures, and that the number of data
structures in use is bounded by some constant. To support implicit
batching, each data structure call is automatically handled by a \textbf{parallel
buffer} for the data structure. (See Appendix \ref{sub:par-buffer}.)

The execution DAG $E$ consists of \textbf{core nodes} and \textbf{ds
nodes}, which are dynamically generated as follows. At the start $E$
has a single core node, corresponding to the start of the program
$P$. Each node could be a \textbf{local instruction} or a \textbf{synchronization
instruction} (including fork/join and acquire/release of a lock).
Each core node could also be a data structure call. When a node is
executed, it may generate child nodes or terminate. A join instruction
also generates edges that linearize all the join operations according
to the actual execution. Likewise, simultaneous operations on a non-blocking
lock generate child nodes that are linearized by edges. For a blocking
lock, a release instruction generates a child node that is simply
the \textbf{resumed thread} that next acquires the lock (if any),
with an edge to it from the node corresponding to the originally suspended
thread.

The core nodes are further classified into \textbf{program nodes}
and \textbf{buffer nodes}. The program nodes (here \textbf{$P$-nodes})
correspond to nodes in the program DAG $D$, and they generate only
program nodes except for data structure calls. An $M$-call generates
a \textbf{buffer node} corresponding to passing the call to the parallel
buffer. This buffer node generates more buffer nodes, until at some
point it generates an $M$-node (every \textbf{$M$-node} is a ds
node), corresponding to the actual operation on $M$, which passes
the \textbf{input batch} to $M$. That $M$-node generates only $M$-nodes
except for when it returns the result of some operation in the batch
(generating a program node with an edge to it from the original $M$-call),
or when it becomes ready for input (generating a buffer node that
initiates flushing of the parallel buffer).

\subsection*{Effective Cost}

We shall now precisely define the notion of effective work/span/cost
for a parallel data structure used by a (terminating) parallel program.
\begin{defn}[Effective Work/Span/Cost]
\label{def:effective} Take any program $P$ using a batched data
structure $M$ on $p$ processors. Let $E$ be the actual execution
DAG of $P$ using $M$. Then the \textbf{effective work} taken by
$M$ (as used by $P$) is the total number $w$ of $M$-nodes in $E$.
And the \textbf{effective span} taken by $M$ is the maximum number
$v$ of $M$-nodes on a path in $E$. And the \textbf{effective cost}
of $M$ is $\frac{w}{p}+v$.
\end{defn}
The effective cost has the desired property that it is \textbf{subadditive}
across multiple parallel data structures. This implies that our results
are \textbf{composable} with other data structures in this model,
since we actually show the following for some linearization $L$:
\begin{itemize}
\item (\ref{thm:M1-work} and \ref{thm:M1-span}) $M_{1}$ takes $O\left(W_{L}+e_{L}\cdot\log p\right)$
effective work and $O\left(\frac{W_{L}}{p}+d\cdot\left((\log p)^{2}+\log n\right)\right)$
effective span (using any scheduler).
\item (\ref{thm:M2-work} and \ref{thm:M2-span}) $M_{2}$ takes $O\left(W_{L}+e_{L}\cdot\log p\right)$
effective work and $O\left(\frac{W_{L}}{p}+d\cdot(\log p)^{2}+s_{L}\right)$
effective span (using a weak-priority scheduler (\ref{sub:weak-priority-sched})).
\end{itemize}
Interestingly, the bound for the effective cost of $M_{1}$ is independent
of the scheduler, while the effective cost bound for $M_{2}$ requires
a weak-priority scheduler. In addition, using any greedy scheduler,
the parallel buffer for either map $M$ has effective cost (analogously
defined) at most $O\left(\frac{T_{1}+w_{M}}{p}+d\cdot\log p\right)$
where $w_{M}$ is the effective work taken by $M$ (Appendix \ref{rem:par-buff-cost}).
Therefore our main results (\ref{thm:M1-perf} and \ref{thm:M2-perf})
follow from the above claims.

\subsection*{Memory Model \phantomsection}

\label{sub:mem-model}

Unless otherwise stated, we work within the \textbf{pointer machine
model} for parallel programs given by Goodrich and Kosaraju~\cite{goodrich1996parallelsort}~\footnote{In short, the main memory can be accessed only via pointers, which
can only be stored, dereferenced or tested for equality (no pointer
arithmetic). }. But instead of having synchronous processors, we introduce a new
more realistic~\footnote{Exclusive reads/writes (EREW) is too strict, while concurrent reads/writes
(CRCW) does not realistically model the cost of contention, as stated
in \cite{gibbons1998qrqw}.} \textbf{QRMW} model with queued read-modify-write operations (including
read, write, test-and-set, fetch-and-add, compare-and-swap) as described
in \cite{dwork1997contention}, where multiple memory requests to
the same memory cell are FIFO-queued and serviced one at a time, and
the processor making each memory request is blocked until the request
has been serviced. Our data structures can hence be implemented and
used in the dynamic multithreading paradigm.

This QRMW pointer machine model supports \textbf{binary fork/join}
primitives. It cannot support constant-time random-access locks, but
it supports \textbf{non-blocking locks} (try-locks), where attempts
to acquire the lock are serialized but do not block. Acquiring a non-blocking
lock succeeds if the lock is not currently held but fails otherwise,
and releasing always succeeds. If $k$ threads concurrently access
a non-blocking lock, then each access completes within $O(k)$ time
steps. Non-blocking locks can be used to support \textbf{activation
calls} to a process, where \textbf{activating} a process will start
its execution iff it is not already executing and it is ready (some
condition is satisfied), and the process can optionally reactivate
itself on finishing.

We can also implement a \textbf{dedicated lock}, which is a blocking
lock initialized with keys $[1..k]$ for some constant $k$, such
that simultaneous acquisitions must be made using distinct keys. When
a thread attempts to acquire a dedicated lock, it is guaranteed to
obtain the lock after at most $O(1)$ other threads that attempt to
acquire the lock at the same time or later.

\section{Amortized Sequential Working-set Map}

\label{sec:ASWM}

In this section we explain the amortized sequential working-set map
$M_{0}$, which is similar to Iacono's working-set structure \cite{iacono2001wstree},
but does not move an accessed item all the way to the front. This
\textbf{\textit{localization of self-adjustment}} is the basis for
parallelizing it as in $M_{2}$.

$M_{0}$ keeps the items in a list with segments $S[0..l]$. Each
segment $S[k]$ has capacity $2^{2^{k}}$ and every segment is full
except perhaps the last. Items in each segment are stored in both
a key-map and a recency-map, each of which is a BBT (balanced binary
tree), sorted by key and by recency respectively. Consider any item
$x$ currently in segment $S[k]$. On a search of $x$, if $k=0$
then $x$ is moved to the front (most recent; i.e. first in the recency-map)
of $S[0]$, otherwise $x$ is moved to the front of $S[k-1]$ and
the last (least recent) item of $S[k-1]$ is shifted to the front
of $S[k]$. On a deletion of $x$, it is removed and for each $i\in[k..l-1]$
the first (most recent) item of $S[i+1]$ is moved to the back of
$S[i]$. On an insertion, the item is added at the back of $S[l]$
(if $S[l]$ is full, then it is added to a new segment $S[l+1]$).

We now prove an abstract lemma about a list with operations and costs
that mimic $M_{0}$. We will later use this same lemma to analyze
$M_{1}$ and $M_{2}$ as well.
\begin{lem}[Working-Set Cost Lemma]
\label{lem:work-set-cost} Take any sequence $L$ of operations on
an abstract list $R$, each of which is a search, insert, delete or
demote (described below), and a constant $d\ge0$, such that the following
hold (where $n$ is the current size of $R$):
\begin{itemize}
\item Searching for an item with rank $q$ in $R$ costs $O(\log q+1)$
and it is pulled forward to within the first $2^{d}\cdot q^{1/2}$
items in $R$.
\item Searching for an item not in $R$ costs $O(\log(n+1)+1)$.
\item Inserting or deleting an item costs $O(\log(n+1)+1)$.
\item Demoting an item in $R$ costs $0$ and pushes it backward in $R$,
but that item subsequently can only be demoted or deleted.
\end{itemize}
Then the total cost of performing $L$ on $R$ is $O(W_{L})$, where
demotions are ignored in computing $W_{L}$ (they are not counted
as accesses).\end{lem}
\begin{proof}
We shall perform the analysis via the accounting method; each operation
on $L$ has a budget according to $W_{L}$, and we must use those
credits to pay for that operation, possibly saving surplus credits
for later. Define the $R$-recency of an item $x$ in $R$ to be the
number of items in $R$ that have been inserted before or pulled forward
past $x$ in $R$ since the last operation on $x$. Clearly, for each
search/insertion/deletion of an item $x$ in $L$, its access rank
(actual recency) is at least the $R$-recency of $x$. Each item in
$R$ has some stored credit, and we shall maintain the invariant that
every item in $R$ with $R$-recency $r$ and stored credit $c$ is
within the first $2^{c+2d+1}+r$ items in $R$ or has been demoted.
The invariant trivially holds at the start.

First we show that, on every operation on an item $x$, the invariant
can be preserved for $x$ itself. For insertion/deletion or unsuccessful
search for $x$, the budget of $\Theta(\log(n+1)+1)$ can pay for
the operation and (for insertion) also pay for the stored credit for
$x$. For successful search for $x$, it is as follows. Let $c$ be
the stored credit and $r$ be the $R$-recency of $x$ before the
operation, and let $q$ be the rank of $x$ in $R$ after that. By
the invariant, $x$ was within the first $2^{c+2d+1}+r$ items in
$R$ before the operation. Also the budget is $\Omega(\log r+1)$.
If $r\ge2^{c+2d+1}$, then $q\le2^{d}\cdot\sqrt{2r}$ and so the budget
can pay for both the operation cost and a new stored credit of $\log\sqrt{2r}$.
If $r\le2^{c+2d+1}$, then $q\le2^{d}\cdot\sqrt{2^{c+2d+1}\cdot2}=2^{c/2+2d+1}$
and so the stored credit can pay for the operation cost and a new
stored credit of $c/2$.

Finally we check that the invariant is preserved for every other item
$y$ in $M$. For search/insertion of $x$, the rank of $y$ in $R$
changes by the same amount as its $R$-recency. For deletion of $x$,
if $x$ is after $y$ in $R$ then $y$ is more recent than $x$ and
so the $R$-recency of $y$ does not change, and if $x$ is before
$y$ in $R$ then the rank of $y$ in $R$ decreases by $1$ and its
$R$-recency decreases by at most $1$. For a demotion, every other
item's rank in $R$ does not increase.
\end{proof}
This lemma implies that $M_{0}$ has the desired working-set property.
\begin{thm}[$M_{0}$ Performance]
\label{thm:M0-perf} The cost of $M_{0}$ satisfies the working-set
bound.\end{thm}
\begin{proof}
Let $n$ be the number of items in $M_{0}$. By construction, $2^{l-1}\le\log(n+1)$,
and each operation on $M_{0}$ takes $O\left(2^{k}\right)$ time on
segment $S[k]$. Thus each insertion/deletion takes $O(\log(n+1)+1)$
time, and each access/update of an item with rank $q$ in $M_{0}$
(in order of segment followed by order in the recency-map) takes $O(\log q+1)$
time. Also, on each access of an item $x$ with rank $q$ in $M_{0}$,
its new rank $q'$ is at most $2q^{1/2}$, because if $x$ is in $S[0..1]$
then $q'=1$, and if $x$ is in $S[k]$ for some $k>1$ then $q\ge2^{2^{k-1}}$
and $q'\le2^{2^{k-2}}\cdot2$. Thus by the \nameref{lem:work-set-cost}
(\ref{lem:work-set-cost}) we are done.
\end{proof}

\section{Simple Parallel Working-Set Map}

\label{sec:PWM1}

We now present our simple batched working-set map $M_{1}$. The idea
is to use the amortized sequential working-set map $M_{0}$ (\ref{sec:ASWM})
and execute operations in batches~\footnote{\label{fn:batch}Each batch is stored in a leaf-based balanced binary
tree for efficient processing.}. In order to get the bound we desire, however, we must combine operations
in a batch that access the same item. In particular, for consecutive
accesses to the same item, all but the first one should cost $O(1)$.
Therefore, we must sort the batch (efficiently) using the \nameref{def:par-esort}
(Appendix \ref{def:par-esort}), to `combine duplicates'. We also
control the size of batches --- if batches are too small, then we
lose parallelism; if the batches are too large, then the sorting cost
is too large.

\subsection{Description of $M_{1}$}

\label{sub:M1-desc}

\begin{figure}[H]
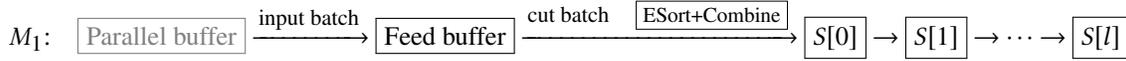

\noindent \begin{centering}
$M_{1}$:\quad{}$\grey{\tbox{Parallel\ buffer}}\xrightarrow{\text{input batch}}\tbox{Feed\ buffer}\xrightarrow{\text{cut batch}\quad\stbox{\text{ESort}+\text{Combine}}}\boxed{S[0]}\to\boxed{S[1]}\to\cdots\to\boxed{S[l]}$
\par\end{centering}

\caption{\label{fig:M1-flow}$M_{1}$ Outline}
\end{figure}

As described in the \ref{sec:model}, $M_{1}$-calls in the program
are put into the parallel buffer for $M_{1}$. When $M_{1}$ is ready
(i.e. the previous batch is done), we flush the parallel buffer to
obtain the next input batch $I$, which we cut and store in a \textbf{feed
buffer}, which is a queue of bunches (a \textbf{bunch} is a set supporting
$O(1)$ addition of a batch and $O(\log b)$-span conversion to a
batch if it has size $b$)~\footnote{\label{fn:bunch}A bunch is implemented using a complete binary tree
with batches at the leaves, with a linked list threaded through each
level to support adding a leaf in $O(1)$ steps.} each of size $p^{2}$ except possibly the last. Specifically, we
divide $I$ into small batches of size $p^{2}$ except possibly the
first and last, the first having size $\min\left(b,p^{2}-q\right)$,
where $b$ is the size of $I$ and $q$ is the size of the last bunch
$F$ in the feed buffer. Next we add that first small batch to $F$,
and append the rest as bunches to the feed buffer. Let $n$ be the
current size of $M_{1}$. Then we remove the first $\ceil{\frac{\log n}{p}}$
bunches from the feed buffer, convert them into batches, merge them
in parallel into a \textbf{cut batch} $B$, and process it as follows.

First we sort $B$ using the \nameref{def:par-esort} (Appendix \ref{def:par-esort}).
Then we combine all operations on each item into one \textbf{group-operation}~\footnote{\label{fn:group-op}Each group-operation stores the original operations
as a batch.} that is treated as a single operation with the same effect as the
whole group of operations in the given order.

We pass the resulting batch (of group-operations) through each segment
from $S[0]$ to $S[l]$. At segment $S[k]$, first we search for the
relevant items. For insertions, if the item is found then we treat
it as an update. For successful searches/updates, we return the results
immediately, and shift the items (keeping their relative order) to
the front of the previous segment $S[k-1]$ ($S[0]$ if $k=0$). For
deletions, if the item is found then we delete the item. Next, we
restore the capacity invariant for $S[0..k-1]$ --- for each segment
$S[i]$ from $S[k]$ to $S[1]$, we transfer the appropriate number
of items between the front of $S[i]$ and the back of $S[i-1]$ so
that either $S[0..i-1]$ have total size $\sum_{j=0}^{i-1}2^{2^{j}}$
or $S[i]$ is empty. Then we pass the unfinished (unreturned) operations
(including all deletions) on to the next segment.

At the end, we handle the remaining insertions. First we insert at
the back of $S[l]$ up to its capacity. If there are leftover items,
we create segments $S[l+1..l']$ with just enough capacity, and carve
out the correct amount for each segment from $S[l']$ to $S[l+1]$
in that order.

Finally, we return the results for the insertions/deletions and the
unsuccessful searches/updates, and we are done processing the batch.

To parallelize the above, we need to be able to efficiently delete,
from each segment, any sorted batch of items or any number of the
most/least recent items. For this we replace each BBT by a \nameref{sub:par-2-3-tree}
(Appendix \ref{sub:par-2-3-tree}), where each leaf in the key-map
also has a direct pointer to the corresponding leaf in the recency-map,
and vice versa. Given a sorted batch of items, we can find them by
one batch operation on the key-map, and then we have a batch of direct
pointers to the leaves for these items in the recency-map, and hence
can perform one reverse-indexing operation on the recency-map to obtain
a sorted batch of indices for those leaves, with which we can perform
one batch operation to delete them. Similarly, to remove the $b$
most recent items from a segment, we find them via the recency-map,
and then do a reverse-indexing operation on the key-map to obtain
them in sorted order, whence we can delete them.

\subsection{Analysis of $M_{1}$}

We first bound the cost of parallel entropy-sorting each batch (\ref{lem:batch-sort-cost}).
To do so, we will find a batch-preserving linearization (\ref{def:batch-preserve})
$L$ such that for each batch of size $b$, the entropy bound is at
most the insert working-set bound (\ref{def:ins-work-set}), which
we in turn bound by the cost according to the working-set bound $W_{L}$,
plus $O(\log b)$ per operation when the map is small (i.e. has size
$b^{O(1)}$). This extra cost arises when a batch has many operations
on distinct items while the map is small, which according to $W_{L}$
are cheap.
\begin{defn}[Batch-Preserving Linearization]
\label{def:batch-preserve} Take any sequence $I$ of batches of
operations on a map $M$. We say that $L$ is a \textbf{batch-preserving
linearization} of $I$ if $L$ is a permutation of the operations
in $I$ that preserves the ordering of batches and (within each batch)
the ordering of operations on the same item.\end{defn}
\begin{rem*}
\label{rem:batch-preserve-inv} When any two batch-preserving linearizations
of $I$ are performed on $M$, the items in $M$ are the same after
each batch, and the successfulness of each operation remain the same.\end{rem*}
\begin{defn}[Insert Working-Set Bound]
\label{def:ins-work-set} The \textbf{insert working-set bound} $IW_{L}$
for any sequence $L$ of map operations is the working-set bound for
`inserting' the items in $L$ in the given order (ignoring the actual
operations) into an empty map, namely for each item first searching
for it and then inserting it iff it is absent.\end{defn}
\begin{lem}[Batch-Sorting Cost Lemma]
\label{lem:batch-sort-cost} Take any sequence $I$ of batches of
operations on a map $M$, and any constant $\e>0$. Then there is
a batch-preserving linearization $L$ of $I$ such that parallel entropy-sorting
(Appendix \ref{def:par-esort}) each batch in $I$ takes $O(W_{L}+\sum_{B\in L}e_{B}\cdot\log b_{B})$
total work over all batches, where each batch $B$ in $L$ has size
$b_{B}$ and has $e_{B}$ operations that are performed when $M$
has size less than ${b_{B}}^{\e}$ (when $L$ is performed on $M$).\end{lem}
\begin{proof}
Let $L$ be a batch-preserving linearization of $I$ such that each
batch $B$ in $L$ has the maximum insert working-set bound $IW_{B}$
(\ref{def:ins-work-set}). By the \nameref{thm:work-set-worst} (Appendix
\ref{thm:work-set-worst}) $IW_{B}$ is at least the entropy bound
for $B$. Thus parallel entropy-sorting $B$ takes $O(IW_{B})$ work
(Appendix \ref{thm:par-esort-perf}).

Let $b$ be the size of $B$, and $u$ be the number of distinct items
(accessed by operations) in $B$. Partition $B$ into subsequences
$B_{0}$ and $B_{1}$ such that $B_{0}$ has only the first operation
of every distinct item in $B$. For each $i\in\{0,1\}$, let $C_{i}$
be the cost of the operations in $B_{i}$ according to $W_{L}$, and
let $IC_{i}$ be the cost of the operations in $B_{i}$ according
to $IW_{B}$, so $IW_{B}=IC_{0}+IC_{1}$. Let $e_{B}$ be the number
of operations in $B$ performed when $M$ has size less than $b^{\e}$
(according to $L$).

Note that $IC_{1}\le C_{1}$ because every operation in $B_{1}$ is
a successful search according to $IW_{B}$ with access rank no more
than according to $W_{L}$. Thus it suffices to show that $IC_{0}\in O(C_{0}+e_{B}\cdot\log b)$.

If $e_{B}>\frac{1}{2}u$, then obviously $IC_{0}\in O(e_{B}\cdot\log b)$.

If $e_{B}\le\frac{1}{2}u$, then at least $\frac{1}{2}u$ operations
in $B_{0}$ are performed when $M$ has size at least $b^{\e}$. So
according to $W_{L}$, each of those operations has access rank at
least $b^{\e}$ and hence costs $\log(b^{\e})\in\Omega(\log u)$.
Thus $C_{0}\in\Omega(u\cdot\log u)$. Also, $IC_{0}\in O(u\cdot\log u)$,
since any insertion on a map with at most $u$ items has access rank
$O(u)$.
\end{proof}

Next we prove a simple lemma that allows us to divide the work done
on the segments among the operations.
\begin{lem}[$M_{1}$ Segment Work]
\label{lem:M1-segment-work} Each segment $S[k]$ takes $O\left(2^{k}\right)$
work per operation that reaches it.\end{lem}
\begin{proof}
Searching/deleting/shifting the relevant items in the parallel 2-3
trees takes $O\left(2^{k}\right)$ work per operation. Also, for each
$i\le k$, the number of transfers (to restore the capacity invariant)
between $S[i-1]$ and $S[i]$ is at most the number of operations,
and each transfer takes $O\left(2^{i}\right)$ work because there
are always at most $\sum_{j=0}^{i}2^{2^{j}}\le2\cdot2^{2^{i}}$ items
in $S[0..i]$. Thus the transfers take $O\left(2^{k}\right)$ total
work.
\end{proof}
Then we can prove the desired effective work bound for $M_{1}$.

\clearpage{}
\begin{thm}[$M_{1}$ Effective Work]
\label{thm:M1-work} $M_{1}$ takes $O(W_{L}+e_{L}\cdot\log p)$
effective work for some linearization $L$ of $D$.\end{thm}
\begin{proof}
Cutting the input batch of size $b$ from the parallel buffer into
small batches takes $O(b)$ work. Adding the first small batch to
the last bunch in the feed buffer takes $O(1)$ work. Inserting the
bunches into the feed buffer takes $O(b)$ work. Forming a cut batch
of size $b'$ (converting the bunches and merging the results) takes
$O(b')$ work. So all this buffering work adds up to $O(1)$ per map
operation.

Sorting the (cut) batches takes $O(W_{L}+e_{L}\cdot\log p)$ total
work (over all batches) for some linearization $L$, by the \nameref{lem:batch-sort-cost}
(\ref{lem:batch-sort-cost}). Specifically, we choose $\e=\frac{1}{3}$.
For each batch $B$, let $n$ be the size of $M_{1}$ just before
that batch, and then $B$ has size $b^{*}\le\ceil{\frac{\log n}{p}}\cdot p^{2}\le p\cdot\log n+p^{2}$
and:
\begin{itemize}
\item If $n\le3p^{2}$, then $b^{*}\le p\cdot\log n+p^{2}\le p\cdot\log(3p^{2})+p^{2}\le4p^{2}$
and so ${b^{*}}^{\e}\le p$ (as $p\ge4$).
\item If $n>3p^{2}$, then $n>p\cdot\log n+p^{2}+p\ge b^{*}+p$ so none
of the operations in that batch can be small-ops.
\end{itemize}
It now suffices to show that the work on segments is $O(W_{L'})$
for some linearization $L'$ (since either $L$ or $L'$ suffices
for the final bound). For this, we pretend that a deleted item is
marked rather than removed, and when a segment is filled to capacity
all marked items are simultaneously transferred to the next segment,
and at the last segment the marked items are removed. This takes more
work than what $M_{1}$ actually does, but is easier to bound.

We shall now use the \nameref{lem:work-set-cost} (\ref{lem:work-set-cost})
on the list $R$ of the items in $M_{1}$ (including the marked items)
in order of segment followed by recency within the segment, where
$R$ is updated after the batch has passed through each segment in
the actual execution $A$ of $M_{1}$, and after we finish processing
the batch.

We simulate the updates to $R$ by list operations as follows:
\begin{itemize}
\item Shift successfully searched/updated items in $A$: Search for them
in reverse order (from back to front in $R$).
\item Shift marked (to-be-deleted) items in $A$: Demote them.
\item Insert items in $A$: Insert in the desired positions.
\item Remove marked items in $A$: Delete them.
\end{itemize}
This simulation yields a sequence $G$ of list operations on $R$,
to which we can then apply the \nameref{lem:work-set-cost}.

For each search for an item $x$ with rank $q$ in $R$, $x$ is found
in $S[0]$ or some segment $S[k+1]$ such that $2^{2^{k}}<q$, and
so by \nameref{lem:M1-segment-work} (\ref{lem:M1-segment-work})
the search takes $O(\log q+1)$ work in $A$, after which $x$ has
new rank in $R$ at most $2q^{1/2}$, like in $M_{0}$ (\ref{thm:M0-perf}).
After each batch $B$, let $n'$ be the final size of $M_{1}$ and
$S[l']$ be the new last segment, and then each insertion in $B$
takes $O\left(\sum_{i=0}^{l'}2^{i}\right)\wi O(\log n'+1)$ work in
$A$. Each deletion takes $O(\log n+1)$ work in $A$.

Thus by \ref{lem:work-set-cost}, $M_{1}$ takes $O(W_{G})$ work
on segments. Now let $L'$ be the same as $G$ but with each group-operation
expanded to its original sequence of operations. Clearly $W_{G}\le W_{L'}$,
since each group-operation is on the same item, so we are done.
\end{proof}
And now we turn to bounding the effective span.
\begin{thm}[$M_{1}$ Effective Span]
\label{thm:M1-span} $M_{1}$ takes $O\left(\frac{N}{p}+d\cdot\left((\log p)^{2}+\log n\right)\right)$
effective span, where $N$ is the number of operations on $M_{1}$,
and $n$ is the maximum size of $M_{1}$.\end{thm}
\begin{proof}
First we bound the the span of processing each cut batch (i.e. the
span of the corresponding execution subDAG). Let $s(b)$ denote the
maximum span of processing a cut batch of size $b$. Take any cut
batch $B$ of size $b$ and let $n_{B}$ be the size of $M_{1}$ just
before $B$. $B$ takes $O\left(\frac{b}{p^{2}}+\log b\right)$ span
to be removed and formed from the feed buffer, and $O\left((\log b)^{2}\right)$
span to be sorted. $B$ then takes $O\left(\log b+2^{k}\right)$ span
in each segment $S[k]$ (because shifting between parallel 2-3 trees
of size $O\left(2^{h}\right)$ or cutting a batch of size $O\left(2^{h}\right)$
takes $O(h)$ span), which adds up to $O\left(\log b\cdot\log\log b+\log n_{B}\right)$
span over all segments, since $\log b<2^{k}$ when $k>\log\log b$.
Returning the results for each group-operation takes $O(\log b)$
span. Thus $s(b)\in O\left(\frac{b}{p}+(\log b)^{2}+\log n_{B}\right)$.
If $b\le p^{2}$ then $s(b)\in O\left(\frac{b}{p}+(\log p)^{2}+\log n_{B}\right)$.
If $b\ge p^{2}$ then $(\log b)^{2}\in O\left(\frac{b}{p}\right)$
and hence $s(b)\in O\left(\frac{b}{p}+\log n_{B}\right)$.

Now let $E$ be the actual execution DAG for $P$ using $M_{1}$ (on
$p$ processors). Then the effective span of $M_{2}$ is simply the
time taken to execute $E$ on an unlimited number of processors when
each $M_{1}$-node in $E$ takes unit time while every other node
takes zero time, since $E$ captures all relevant behaviour of $P$
using $M_{1}$ including all the dependencies created by the locks.
In this execution, we put a counter at each $M_{1}$-call in the program
DAG $D$, initialized to zero, and at each step we increment the counter
at every pending $M_{1}$-call (i.e., the result is not yet returned).
Then the total number of steps is at most the final counter-weighted
span of $D$, which we now bound.

Take any path $C$ in $D$. Consider each $M_{1}$-call $X$ on $C$.
We trace the `journey' of $X$ from the parallel buffer as an operation
in an uncut batch $U$ of size $u$ to a cut batch $B$ of size $b$
to the end of $M_{1}$.

Observe that any batch of size $u$ takes $O(\log p+\log u)$ span
to be flushed from the parallel buffer, and $O\left(\log u+\frac{u}{p^{2}}\right)$
span to be cut and added/appended to the bunches in the feed buffer,
which in total is at most $O\left(\log p+\frac{u}{p}\right)$ span.

So, first of all, $X$ waits for the preceding uncut batch of size
$u'$ to be processed, taking $O\left(\log p+\frac{u'}{p}\right)$
span. Next, $X$ waits for the current cut batch $B'$ of size $b'$
to be processed, taking $s(b')$ span. After that, $U$ is processed,
taking $O\left(\log p+\frac{u}{p}\right)$ span. Then $X$ waits for
intervening cut batches (between $B'$ and $B$) with $i$ operations
in total. Each intervening batch $B^{*}$ has some size $b^{*}\ge\max\left(p^{2},p\cdot\log n_{B^{*}}\right)$
and hence $s(b^{*})\in O\left(\frac{b^{*}}{p}\right)$. Finally, $B$
is processed, taking $s(b)$ span. Thus $X$ takes $O\left(\log p+\frac{u}{p}+\frac{u'}{p}+s(b)+s(b')+\frac{i}{p}\right)$
span in total.

Note that no two $M_{1}$-calls on the path $C$ can wait for the
same intervening batch, because the second can be executed only after
the first has returned. Thus over all counters at $M_{1}$-calls on
$C$, each of $u,u',b,b',i$ will sum up to at most $N$. Therefore
the final counter-weighted span of $D$ is at most $O\left(\frac{N}{p}+d\cdot\left((\log p)^{2}+\log n\right)\right)$.
\end{proof}

\section{Faster Parallel Working-Set Map}

\label{sec:PWM2}

To reduce the effective span of $M_{1}$, we intuitively have to:
\begin{itemize}
\item Shift each accessed item near enough to the front, so that accessing
it again soon would be cheap.
\item Pipeline the batches somehow, so that an expensive access in a batch
does not hold up the next batch.
\end{itemize}
\noindent Naive pipelining will not work, because operations on the
same item may take too much work. Hence we shall use a \textbf{filter}
before the pipelined segments to ensure that operations proceeding
through them are on distinct items, namely we pass all operations
through the filter and only allow an operation through if there is
not already another operation on the same item in the pipeline.

For similar reasons as in $M_{1}$, we must control both the batch
size and filter size to achieve enough parallelism, and so we choose
the filter capacity to be $\Theta\left(p^{2}\right)$. However, we
cannot put the filter before the first segment, because accessing
the filter requires $\Omega(\log p)$ work per operation, whereas
to meet the working-set bound we need operations with $O(1)$ access
rank to cost only $O(1)$ work.

Therefore, we divide the segments into the \textbf{first slab} and
the \textbf{final slab}, where the first slab comprises the first
$\log\Theta(\log p)$ segments and the final slab contains the rest,
and put the filter after the first slab. Only operations that do not
finish in the first slab are passed through the filter, and so the
filtering cost per operation is bounded by the $\Theta(\log p)$ work
already incurred in going through the first slab. Furthermore, we
shift accessed items to the front of the final slab, and `cascade'
the excess items only when a later batch passes.

We cannot pipeline the first slab, but since the first slab is essentially
a copy of $M_{1}$ but with only $\log O(\log p)$ trees, its non-pipelined
span turns out to be bounded by the $O\left((\log p)^{2}\right)$
span of sorting. To allow operation on items in the first slab to
finish quickly, we need to allow the first slab to run while the final
slab is running, but only when the filter has size at most $p^{2}$,
so that the filter size is always $O\left(p^{2}\right)$.

We also use \textbf{special locking schemes} to guarantee that the
first slab and the segments in the final slab can process the operations
at a consistent pace without interfering with one another. Finally,
we shall \textbf{weakly prioritize} the execution of the final slab,
to prevent excessive work from being done in the first slab on an
item $x$ if there is already an operation on $x$ in the final slab.

\noindent

\subsection{Description of $M_{2}$}

\begin{figure}[H]
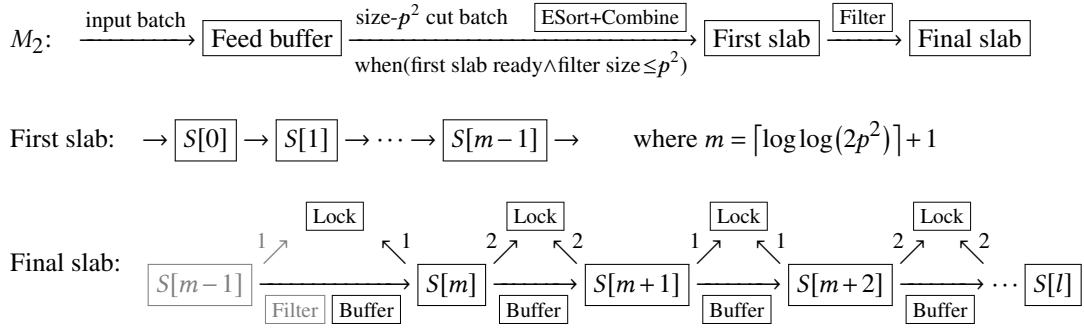

\begin{centerbox}
\noindent %
\begin{tabular}{l}
$M_{2}$:\quad{}$\xrightarrow{\text{input batch}}\tbox{Feed\ buffer}\xrightarrow[\text{when(first slab ready\ensuremath{\land}filter size\ensuremath{\le}\ensuremath{\ensuremath{p^{2}}})}]{\text{size-\ensuremath{p^{2}} cut batch}\quad\stbox{ESort+Combine}}\tbox{First\ slab}\xrightarrow{\stbox{Filter}}\tbox{Final\ slab}$\tabularnewline
\noalign{\vskip1.5em}
First slab:\quad{}$\to\boxed{S[0]}\to\boxed{S[1]}\to\cdots\to\boxed{S[m-1]}\to$\quad{}\quad{}where
$m=\ceil{\log\log\left(2p^{2}\right)}+1$\tabularnewline
\noalign{\vskip1.5em}
Final slab:\quad{}$\begin{matrix}\vsp{.4}\\
\grey{\boxed{S[m-1]}}
\end{matrix}\begin{matrix}\stbox{Lock}\\
\grey{\mathrlap{^{1}}\nearrow}\qquad\quad\ \nwarrow\mathllap{^{1}}\\
{}\xrightarrow[\grey{\stbox{Filter}}\ \stbox{Buffer}]{}{}
\end{matrix}\begin{matrix}\vsp{.4}\\
\boxed{S[m]}
\end{matrix}\begin{matrix}\stbox{Lock}\\
\mathrlap{^{2}}\nearrow\quad\ \nwarrow\mathllap{^{2}}\\
{}\xrightarrow[\stbox{Buffer}]{}{}
\end{matrix}\begin{matrix}\vsp{.4}\\
\boxed{S[m+1]}
\end{matrix}\begin{matrix}\stbox{Lock}\\
\mathrlap{^{1}}\nearrow\quad\ \nwarrow\mathllap{^{1}}\\
{}\xrightarrow[\stbox{Buffer}]{}{}
\end{matrix}\begin{matrix}\vsp{.4}\\
\boxed{S[m+2]}
\end{matrix}\begin{matrix}\stbox{Lock}\\
\mathrlap{^{2}}\nearrow\quad\ \nwarrow\mathllap{^{2}}\\
{}\xrightarrow[\stbox{Buffer}]{}{}
\end{matrix}\begin{matrix}\vsp{.4}\\
\cdots\,\boxed{S[l]}
\end{matrix}$\tabularnewline
\noalign{\vskip1em}
\end{tabular}
\end{centerbox}
\caption{\label{fig:M2-flow}$M_{2}$ Outline}
\end{figure}

We shall now give the details for implementing this (see \ref{fig:M2-flow}),
making considerable use of the \nameref{sub:par-2-3-tree} (Appendix
\ref{sub:par-2-3-tree}). $M_{2}$ has the same segments as in $M_{1}$,
where segment $S[k]$ has assigned capacity $2^{2^{k}}$ but may be
under-full or over-full. We shall group the first $m=\ceil{\log\log\left(2p^{2}\right)}+1$
segments into the \textbf{first slab}, and the other segments into
the \textbf{final slab}. $M_{2}$ uses a \textbf{feed buffer} (like
$M_{1}$; see \ref{sub:M1-desc}), which is a queue of bunches~\ref{fn:bunch}
each of size $p^{2}$ except possibly the last.

The $M_{2}$ interface is ready iff both the following hold:
\begin{itemize}
\item The parallel buffer or feed buffer is non-empty.
\item The filter has size at most $p^{2}$.
\end{itemize}
When the $M_{2}$ interface is activated (and ready), it does the
following (in sequence) on its run (the locks are described later):
\begin{enumerate}
\item Let $q$ be the size of the last bunch $F$ in the feed buffer. Flush
the parallel buffer and cut the input batch of size $b$ into small
batches of size $p^{2}$ except possibly the first and last, where
the first has size $\min\left(b,p^{2}-q\right)$. Add that first small
batch to $F$, and append the others as bunches to the feed buffer.
Remove the first bunch from the feed buffer and convert it into a
batch $B$, which we shall call a \textbf{cut batch}.
\item Sort $B$ using the \nameref{def:par-esort} (Appendix \ref{def:par-esort}),
combining operations on the same item~\ref{fn:group-op}, as in $M_{1}$.
\item \label{enu:first-slab-run} Pass $B$ through the first slab, which
processes the operations as in $M_{1}$. Successful searches/updates
immediately finish, while the rest finish only if there was no final
slab. Successful deletions are tagged to indicate success. But just
before running $S[m-1]$ (if it exists) to process the remaining batch
at that segment, acquire the neighbour-lock shared with $S[m]$ (as
shown in \ref{fig:M2-flow}) and then \textcolor{magenta}{acquire
the front-lock $FL[0]$}. 
\item If there was a final slab, then pass the (sorted) batch of unfinished
operations through the filter (including successful deletions), insert
the filtered batch into the buffer before $S[m]$, and fork (a child
thread) to activate $S[m]$.
\item \textcolor{magenta}{Release $FL[0]$} and the neighbour-lock shared
with $S[m]$.
\item Reactivate itself.
\end{enumerate}
The \textbf{filter} is used to ensure that at any point all the operations
in the final slab are on distinct items. It is implemented using a
batched parallel 2-3 tree that stores items, each tagged with a list
of operations on that item (in the order they arrive at the filter)
and their cumulative \textbf{effect} (as a single equivalent map operation).

When a batch is passed through the filter, each operation on an item
in the filter is appended to the list for it (the effect is also updated)
and filtered out of the batch, whereas each operation on an item not
in the filter is added to the filter and put into the buffer of $S[m]$.

The final slab is pipelined in the following way. Between every pair
of consecutive segments is a \textbf{neighbour-lock}, which is a dedicated
lock (see \ref{sub:mem-model}\hyperref[sub:mem-model]{\ Memory Model})
with $1$ key for each arrow to it in \ref{fig:M2-flow}. Since each
segment needs to access the filter and the contents of $S[m]$, those
accesses will also be guarded by a front-locking scheme using a series
of \textbf{front-locks} $FL[0..l-m]$, each of which is a dedicated
lock with $1$ key for each arrow to it in \ref{fig:filter-lock}.
(This will be fully spelt out below.)

\begin{figure}[H]
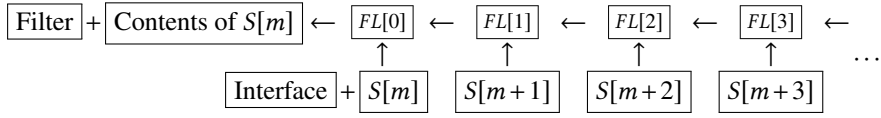

\noindent \begin{centering}
$\begin{matrix}\tbox{Filter}+\boxed{\text{Contents of \ensuremath{S[m]}}}\\
\\
\vphantom{\boxed{X'}}
\end{matrix}\ \begin{matrix}\leftarrow\\
\vphantom{\boxed{X'}}\\
\\
\end{matrix}\begin{matrix}\smbox{FL[0]}\\
\uparrow\\
\mathllap{\tbox{Interface}}+\boxed{S[m]}
\end{matrix}\begin{matrix}\leftarrow\\
\vphantom{\boxed{X'}}\\
\\
\end{matrix}\begin{matrix}\smbox{FL[1]}\\
\uparrow\\
\boxed{S[m+1]}
\end{matrix}\begin{matrix}\leftarrow\\
\vphantom{\boxed{X'}}\\
\\
\end{matrix}\begin{matrix}\smbox{FL[2]}\\
\uparrow\\
\boxed{S[m+2]}
\end{matrix}\begin{matrix}\leftarrow\\
\vphantom{\boxed{X'}}\\
\\
\end{matrix}\begin{matrix}\smbox{FL[3]}\\
\uparrow\\
\boxed{S[m+3]}
\end{matrix}\begin{matrix}\leftarrow\\
\vphantom{\boxed{X'}}\\
\\
\end{matrix}\cdots$
\par\end{centering}

\caption{\label{fig:filter-lock}Front-locking}
\end{figure}

Each final slab segment $S[k]$ has a \textbf{sorted buffer} before
it (for operations from $S[k-1]$), which is a batched parallel 2-3
tree. $S[k]$ is ready iff its buffer is non-empty, and when activated
(and ready) it runs as follows (front-locking is \textcolor{magenta}{highlighted}):
\begin{enumerate}
\item Acquire the neighbour-locks (between $S[k]$ and its neighbours) in
the order given by the arrow number in \ref{fig:M2-flow}.
\item \textcolor{magenta}{If $k=m$, acquire $FL[0]$.}
\item If $S[k]$ is the terminal segment and $S[k-1,k]$ have total size
exceeding their total capacity, create a new terminal segment $S[k+1]$.
\item Flush and process the operations in its buffer as follows:

\begin{enumerate}
\item Search for the accessed items $A$ in $S[k]$ (by performing one batch
operation on the key-map in $S[k]$). Let $R$ be the (sorted batch
of) items in $A$ that are found in $S[k]$, and delete $R$ from
$S[k]$.
\item \textcolor{magenta}{If $k>m$, acquire $FL[k-m]$ to $FL[0]$ in that
order.}
\item Search for $R$ in the filter to determine what to do with it. Let
$R'$ be the items in $R$ to be searched or updated, and delete $R'$
from the filter. (Insertions on items in $R$ are treated as updates.)
Perform all the updates on items in $R$.
\item \label{enu:finish-ops}Let $m'=\min(k-1,m)$. Fork to return the results
for operations on $R'$, and insert $R'$ at the front of $S[m']$.
If $S[k]$ is (now) the terminal segment, perform all insertions at
the front of $S[m']$, and delete $A\less R'$ from the filter, and
fork to return the results for operations on $A\less R'$.
\item If the filter size is at most $p^{2}$, fork to activate the interface.
\item \textcolor{magenta}{If $k>m$, release $FL[0]$ to $FL[k-m]$ in that
order.}
\item \label{enu:overfull-transfer} If $S[k-1]$ is over-full, transfer
items from the back of $S[k-1]$ to the front of $S[k]$ so that $S[k-1]$
is full.
\item \label{enu:underfull-transfer} If $S[k-1]$ is under-full by $i$
items and $S[k]$ has $c$ items and $A$ has $d$ successful deletions,
transfer $\min(i,c,d)$ items from the front of $S[k]$ to the back
of $S[k-1]$.
\item If $S[k]$ is not the terminal segment, insert the operations on $A\less R'$
(with successful deletions tagged as such) into the buffer of $S[k+1]$,
then fork to activate $S[k+1]$.
\end{enumerate}
\item If $S[k]$ is the terminal segment and is empty, remove $S[k]$ to
make $S[k-1]$ the new terminal segment.
\item \textcolor{magenta}{If $k=m$, release $FL[0]$.}
\item Release both neighbour-locks and reactivate itself.
\end{enumerate}

\subsection{Weak-priority scheduler}

\label{sub:weak-priority-sched}

It turns out that, ignoring the sorting cost, all we need to achieve
the working-set bound is that, for each operation on an item $x$
in the final slab, the work `done' by the first slab on $x$ can
be counted against the work done in the final slab. This can be ensured
using a \textbf{weak-priority scheduler}. A weak-priority scheduler
has two queues $Q_{1}$ and $Q_{2}$, where $Q_{1}$ is the high-priority
queue, and each ready node is assigned to either $Q_{1}$ or $Q_{2}$,
and at every step the following hold:
\begin{itemize}
\item If there are $k$ ready nodes, then $\min(k,\frac{1}{2}p)$ of them
are executed.
\item If queue $Q_{1}$ has $k$ ready nodes, then $\min\left(k,\frac{1}{2}p\right)$
of them are executed (and so $Q_{1}$ nodes are weakly-prioritized).
\end{itemize}
The $M_{2}$-nodes generated for the final slab are assigned to $Q_{1}$,
while all other $M_{2}$-nodes are assigned to $Q_{2}$. Specifically,
each (forked) activation call to $S[m]$ and all nodes generated by
that are assigned to $Q_{1}$, except for activation calls to the
$M_{2}$ interface (which are assigned to $Q_{2}$). Any suspended
thread is put back into its original queue when it is resumed (i.e.
the resuming node is assigned to the same queue as the acquiring node).

\subsection{Analysis of $M_{2}$}

For each computation (subDAG of the actual execution DAG), we shall
define its \textbf{delay}, which intuitively captures the minimum
possible time it needs, including all waiting on locks. Each blocked
acquire of a dedicated lock corresponds to an \textbf{acquire-stall
node} $\alpha$ in the execution DAG whose child node $\rho$ is created
by the release just before the successful acquisition of the lock.
Let $\Delta(\alpha)$ be the ancestor nodes of $\rho$ that have not
yet executed at the point when $\alpha$ is executed. Then we define
delay as follows.
\begin{defn}[Computation Delay]
\label{def:comp-delay} The \textbf{delay} of a computation $\Gamma$
is recursively defined as the weighted span of $\Gamma$, where each
acquire-stall node $\alpha$ in $\Gamma$ is weighted by the delay
of $\Delta(\alpha)$ (to capture the total waiting at $\alpha$),
and every other node has unit weight.
\end{defn}
Also, we classify each operation as \textbf{segment-bound} iff it
finishes in some segment, and \textbf{filter-bound} otherwise (namely
filtered out due to a prior operation on the same item that finishes
in the final slab).

The \textbf{\textit{key lemma}} (\ref{lem:M2-final-slab-bound}) is
that any operation that finishes in segment $S[m+k]$ takes $O\left(2^{m+k}\right)$
delay to be processed by the final slab, which is ensured (\ref{lem:M2-front-access-bound})
by the front-locking scheme and the balance invariants (\ref{lem:M2-balance}).

Then $M_{2}$ takes $O(W_{L}+e_{L}\cdot\log p)$ effective work for
some linearization $L$ of $D$, and here is the intuitive proof sketch:
\begin{enumerate}
\item The work divides into \textbf{segment work} (done by the slabs) and
\textbf{non-segment work} (cutting and sorting batches, and filtering).
\item The segment work can be divided per operation; segment $S[k]$ does
$O\left(2^{k}\right)$ work per operation that reaches it (\ref{lem:M2-segment-work}).
\item The cutting work is $O(1)$ per operation, and the (overall) sorting
work is $O(W_{L}+e_{L}\cdot\log p)$ for some linearization $L$ of
$D$ (\ref{lem:batch-sort-cost}). The filtering work is $O(\log p)$
per filtered operation, which can be ignored since each filtered operation
already took $\Omega(\log p)$ work in the first slab. Similarly we
can ignore the work done in passing the batch through $S[m-1]$.
\item The segment work on segment-bound operations is $O(W_{L'})$ for some
linearization $L'$ of $D$ (\ref{lem:M2-core-work}; proven using
\ref{lem:work-set-cost} like we did for $M_{1}$).
\item \br The segment work on filter-bound operations is $O(W_{L'})$ too:

\begin{enumerate}
\item We can ignore work during \textbf{high-busy steps} (where $Q_{1}$
has at least $\frac{1}{2}p$ ready nodes), because the final slab
takes $O(W_{L'})$ work and so there are $O\left(\frac{W_{L'}}{p}\right)$
high-busy steps.
\item We can ignore every \textbf{high-busy run} of $S[0..m-2]$ (namely
with at least half its steps high-busy), because its work is $O(1)$
times the work during high-busy steps.
\item High-idle runs of $S[0..m-2]$ take $O(W_{L'})$ total work.

\begin{enumerate}
\item Every filter-bound operation is filtered out due to some operation
in the final slab. So take any operation $X$ on item $x$ that finishes
in a final slab segment $S[k]$.
\item During each high-idle run (which takes $\Omega(\log p)$ high-idle
steps), the processing of $X$ is not blocked by any $Q_{2}$-thread
(since $S[0..m-2]$ does not hold any neighbour-lock or filter-lock),
so each high-idle step `reduces' its remaining delay, which by the
\textit{key lemma} (\ref{lem:M2-final-slab-bound}) is $O\left(2^{k}\right)$.
\item Therefore the work done on $x$ by high-idle runs while $X$ is in
the final slab is $O(1)$ times the work done by the final slab on
$X$.
\end{enumerate}
\end{enumerate}
\end{enumerate}
Moreover, $M_{2}$ takes $O\left(\frac{W_{L}}{p}+d\cdot(\log p)^{2}+s_{L}\right)$
effective span for some linearization $L$ of $D$:
\begin{enumerate}
\item The effective span is the time taken to run the execution DAG on infinitely
many processors, where each $M_{2}$-node takes unit time while every
other node takes zero time.
\item There are $O\left(\frac{W_{L'}}{p}\right)$ \textbf{filter-full steps}
(steps in which the filter has size at least $p$). To see why, let
every operation in the final slab consume a token on each step. Then
each filter-full step consumes at least $p$ tokens. But by the \textit{key
lemma} (\ref{lem:M2-final-slab-bound}) the total token consumption
is just $O(1)$ times the total work in the final slab, which amounts
to $O(W_{L'})$ (\ref{lem:M2-core-work}).
\item There are $O\left(\frac{N}{p}+d\cdot(\log p)^{2}+s_{L^{*}}\right)$
\textbf{filter-empty steps} (filter size at most $p$) for some linearization
$L^{*}$ of $D$, where $N$ is the number of $M_{2}$-calls, because
each operation essentially has the following path:

\begin{enumerate}
\item It waits $O\left((\log p)^{2}\right)$ filter-empty steps for the
current cut batch in the first slab.
\item Then it waits $O\left(\frac{b^{*}}{p}\right)$ filter-empty steps
per intervening cut batch of size $b^{*}=p^{2}$.
\item Finally it takes $O\left((\log p)^{2}+\log r\right)$ steps to pass
through the slabs where $r$ is its access rank according to $L^{*}$,
by the \textit{key lemma} and the rank invariant (\ref{lem:M2-rank-inv}).
\end{enumerate}
\end{enumerate}
We shall now give the details. For the purpose of our analysis, we
consider a segment to be \textbf{running} iff it has acquired all
neighbour-locks and has not released them. We also consider an operation
to be \textbf{in} segment $S[k]$ exactly when it is in the buffer
of $S[k]$ or is being processed by $S[k]$ up to \step\ref{enu:underfull-transfer}
(inclusive). Likewise, we consider an item that is found in $S[k]$
and to be searched/updated to remain in $S[k]$ until it is \textbf{shifted}
to $S[m']$ (in \step\ref{enu:finish-ops}).

We begin by showing that $M_{2}$ remains `balanced'; each non-terminal
segment has size not too far from capacity.
\begin{defn}[$M_{2}$ Segment Holes]
\label{def:M2-hole} We say that a segment $S[k]$ has $c$ \textbf{holes}
if $S[k]$ is not the terminal segment but has $c$ fewer items than
its capacity. (If $S[k]$ exceeds capacity then it has no holes.)\end{defn}
\begin{lem}[$M_{2}$ Balance Invariants]
\label{lem:M2-balance} The following balance invariants hold:
\begin{enumerate}
\item If $S[m]$ is not running or does not exist, then $S[m-1]$ does not
exceed capacity (if it exists).
\item If the interface is not running, $S[0..m-2]$ has no holes, and $S[m-1]$
has at most $d$ holes where $d$ is the number of successful deletions
in $S[m]$.
\item Each final slab segment $S[k]$ has at most $3\cdot2^{2^{k}}$ items.
\item If a final slab segment $S[k]$ is not running, $S[0..k-1]$ is at
most $2p^{2}$ below capacity.
\end{enumerate}
\end{lem}
\begin{proof}
Let $f$ be the filter size. Then $f\le2p^{2}$ always, since the
interface only runs if $f\le p^{2}$, on a batch of size at most $p^{2}$.

\br

\uline{Invariant 1}

$S[m-1]$ can only exceed capacity when $S[m]$ runs, and $S[m]$
restores the invariant (in \step\ref{enu:overfull-transfer}) before
it finishes running.

\br

\uline{Invariant 2}

Only the interface creates more holes in $S[0..m-1]$ (in \step\ref{enu:first-slab-run}),
each corresponding to a unique successful deletion that is inserted
into the buffer of $S[m]$, and so just after the interface finishes
running $S[0..m-1]$ has at most $d$ holes where $d$ is the number
of successful deletions in $S[m]$, and all the holes must be in $S[m-1]$
since $d\le2p^{2}\le2^{2^{m-1}}$. Once $S[m]$ runs, $S[m-1]$ will
have no holes, because either $S[m]$ was the terminal segment or
$S[m]$ had at least $2^{2^{m}}-2p^{2}\ge d$ items by Invariant~4.

\br

\uline{Invariants 3,4}

To establish Invariants~3,4, we shall prove sharper invariants. Let
$e(k)$ be the total size of $S[0..k-1]$ minus their total capacity.
Let $u(k)$ be the number of unfinished operations in $S[k+1..l]$.
Let $d(k)$ be the number of successful deletions in $S[m..k]$. Then
the following invariants hold:
\begin{enumerate}
\item[(A)]  If a final slab segment $S[k]$ is not running, $e(k)+u(k)\le2p^{2}$.
\item[(B)]  For each final slab segment $S[k]$, we always have $e(k)\le4p^{2}$.
\item[(C)]  $e(l+1)\le2p^{2}$.
\item[(D)]  If a final slab segment $S[k]$ is not running, $e(k)+d(k)\ge0$.
\end{enumerate}
Firstly, (A) holds for $S[m]$, because $u(m)\le f\le2p^{2}$ and
$e(m)\le0$ by Invariant 1 since the interface does not insert any
item. Thus by induction it suffices to show that (A) holds for $S[k]$
where $k>m$ assuming (A) holds for $S[k-1]$. This can be done by
the following observations:
\begin{itemize}
\item When $S[k]$ is not running, $e(k)+u(k)$ never increases, because
each search/update/deletion in $S[0..k-1]$ does not increase $e(k)$
or affect $u(k)$, and each operation that finishes in $S[k+1..l]$
increases $e(k)$ by at most $1$ but decreases $u(k)$ by $1$.
\item When $S[k]$ is newly created, $e(k)\le2p^{2}$ by (C) since $S[k-1]$
was the previous terminal segment, and $u(k)=0$.
\item When $S[k]$ is running, $S[k-1]$ is not running. Thus just after
$S[k]$ finishes running, $e(k)\le e(k-1)$ since $S[k-1]$ does not
exceed capacity (due to \step\ref{enu:overfull-transfer}), and $u(k)=u(k-1)$
since $S[k]$ has an empty buffer. Thus $e(k)+u(k)\le e(k-1)+u(k-1)\le2p^{2}$
by (A) for $S[k-1]$.
\end{itemize}
We now establish (B) using (A). Consider each final slab segment $S[k]$
run. Just before that run, $e(k)\le2p^{2}$ by (A), and there are
at most $2p^{2}$ unfinished operations in $S[k..l]$. During that
run, no new operation enters $S[k..l]$, and $e(k)$ increases by
at most $1$ for each unfinished operation in $S[k..l]$ that finishes.
Thus $e(k)\le4p^{2}$ throughout that run.

Next we establish (C). Note that the terminal segment is only changed
by the interface or the previous terminal segment, and so when the
terminal segment $S[l]$ is not running, $e(l+1)$ never increases
because no insertions finish. It suffices to observe the following:
\begin{itemize}
\item Just after $S[l]$ is newly created, it does not exceed capacity and
so $e(l+1)\le e(l)\le2p^{2}$.
\item Just before $S[l+1]$ was removed (making $S[l]$ the new terminal
segment), $S[l+1]$ had finished running, and so $S[l]$ did not exceed
capacity, and hence $e(l+1)\le e(l)\le2p^{2}$ by (A) for $S[l]$.
\item Whenever $S[l]$ runs and does not create a new terminal segment,
just before that run $S[l-1,l]$ do not exceed total capacity and
so $e(l+1)\le e(l-1)$ at that point. There are two cases:

\begin{itemize}
\item If $l=m$: Just before that run, $e(l-1)\le0$ since $S[0..m-2]$
do not exceed capacity, and $S[l]$ has at most $2p^{2}$ operations.
During that run, $e(l+1)$ increases by at most $1$ per unfinished
operation in $S[l]$, and hence after that run $e(l+1)\le2p^{2}$.
\item If $l>m$: Just before that run, $e(l-1)+u(l-1)\le2p^{2}$ by (A)
for $S[l-1]$, and $S[l]$ has $u(l-1)$ operations. During that run,
$e(l+1)$ is increased only by $1$ per unfinished operation in $S[l]$,
and hence after that run $e(l+1)\le2p^{2}$.
\end{itemize}
\end{itemize}
Finally we establish (D). Firstly, (D) holds for $S[m]$ by Invariant~2.
Thus by induction it suffices to show that (D) holds for $S[k]$ where
$k>m$ assuming (D) holds for $S[k-1]$. This can be done by the following
observations:
\begin{itemize}
\item Any search/update/insertion that finishes does not decrease $e(k)$.
\item If $S[k]$ is not running, any deletion that succeeds in $S[0..k-1]$
decreases $e(k)$ by $1$ but increases $d(k)$ by $1$.
\item When $S[k]$ is newly created by an $S[k-1]$ run, after that run
$S[k-2]$ does not exceed capacity (due to \step\ref{enu:overfull-transfer})
and so for each hole in $S[k-1]$ there will be at least $1$ successful
deletion in $S[k]$, and hence $e(k)+d(k)\ge e(k-1)+d(k-1)\ge0$ by
(D) for $S[k-1]$.
\item When $S[k]$ runs, $e(k)+d(k)\ge0$ after the run because:

\begin{itemize}
\item If after the run $S[k-1]$ is exactly full, then at that point $e(k)=e(k-1)$
and $d(k)\ge d(k-1)$, and $e(k-1)+d(k-1)\ge0$ by (D) for $S[k-1]$.
\item If after the run $S[k-1]$ is below capacity and $S[k]$ still exists,
the run must have made $d'$ frontward transfers (in \step\ref{enu:underfull-transfer})
where $d'$ was the number of successful deletions in $S[k]$ at the
start of that run. Thus the run increased $e(k)$ by at least $d'$,
and decreased $d(k)$ by $d'$.
\end{itemize}
\end{itemize}
Finally we can establish Invariants~3,4. By both (B) and (C), for
each final slab segment $S[k]$ we have $e(k+1)\le4p^{2}$, and hence
$S[k]$ has size at most $\sum_{i=0}^{k}2^{2^{i}}+4p^{2}\le2\cdot2^{2^{k}}+2^{2^{m}}\le3\cdot2^{2^{k}}$.
By (D), if a final slab segment $S[k]$ is not running, then $-e(k)\le d(k)\le f\le2p^{2}$
and hence $S[0..k-1]$ is at most $2p^{2}$ below capacity.\end{proof}
\begin{cor}[$M_{2}$ Segment Access Bound]
\label{cor:M2-segment-access-bound} Each batch operation on a parallel
2-3 tree in segment $S[k]$ where $k\ge m-1$ takes $\Theta\left(2^{k}\right)$
work per operation in the batch and $\Theta\left(\log p^{2}+2^{k}\right)\wi\Theta\left(2^{k}\right)$
span.
\end{cor}
Now we can prove a delay bound on the `front access' (through the
front-locks) by each final slab segment.
\begin{lem}[$M_{2}$ Front Access Bound]
\label{lem:M2-front-access-bound} Any segment $S[m+k]$ takes $O\left(2^{m+k}\right)$
total delay to acquire the front-locks $FL[0..k]$ and run the front-locked
section (in-between) and then release $FL[0..k]$. And similarly the
interface takes $O\left(2^{m}\right)$ delay to acquire $FL[0]$ and
run the front-locked section and then release $FL[0]$.\end{lem}
\begin{proof}
The front-locked section takes $O\left(2^{m}\right)$ delay, since
each operation on a parallel 2-3 tree in $S[m-1,m]$ takes $O\left(2^{m}\right)$
span by \ref{cor:M2-segment-access-bound}. We shall show by induction
that any segment $S[i]$ that has acquired $FL[k]$ will release $FL[k]$
within $c\cdot2^{m+k}$ delay, where $c$ is a constant chosen to
make it true when $k=0$. If $k>0$, then $S[i]$ next attempts to
acquire $FL[k-1]$, and if it fails then $S[m+k-1]$ must now be holding
it and will release it within $c\cdot2^{m+k-1}$ delay by induction,
and then $S[i]$ will actually acquire $FL[k-1]$ and then will release
$FL[k]$ within $c\cdot2^{m+k-1}$ delay by induction, which in total
amounts to $c\cdot2^{m+k}$ delay. Therefore any segment $S[m+k]$
that attempts to acquire $FL[k]$ will wait at most $c\cdot2^{m+k}$
delay for any current holder of $FL[k]$ to release it, and then take
at most $c\cdot2^{m+k}$ delay to run its front-locked section and
release $FL[0..k]$, which is in total $O\left(2^{m+k}\right)$ delay.
Similarly for when the interface attempts to acquire $FL[0]$.
\end{proof}
Then we can prove the \textit{key lemma}:
\begin{lem}[$M_{2}$ Final Slab Bound]
\label{lem:M2-final-slab-bound} Take any segment $S[k]$ where $k\ge m$,
and any operation $X$. Then $S[k]$ runs within $O\left(2^{k}\right)$
delay. So if $X$ finishes in segment $S[k]$ then the processing
of $X$ in the final slab takes $O\left(2^{k}\right)$ delay.\end{lem}
\begin{proof}
Once any $S[k]$ acquires the second neighbour-lock, it will finish
within $O\left(2^{k}\right)$ delay, since the operations on the parallel
2-3 trees takes $O\left(2^{k}\right)$ span by \ref{cor:M2-segment-access-bound},
the front access take $O\left(2^{k}\right)$ delay by \ref{lem:M2-front-access-bound},
and inserting the unfinished operations into the buffer of $S[k+1]$
takes $O(\log p)$ span. Thus once any $S[k]$ acquires the first
lock, it waits $O\left(2^{k}\right)$ delay for the holder of the
second lock to finish, and then itself finishes within $O\left(2^{k}\right)$
delay. And once the interface acquires the lock shared with $S[m]$,
it will finish within $O\left(2^{m}\right)$ delay by \ref{lem:M2-front-access-bound}.
Thus any $S[k]$ when run will acquire both locks within $O\left(2^{k}\right)$
delay, and then itself finish within $O\left(2^{k}\right)$ delay.
Therefore, the final slab takes $O\left(\sum_{i=m}^{k}2^{i}\right)\wi O\left(2^{k}\right)$
delay to process any operation that finishes in $S[k]$.
\end{proof}
To bound the total work, we begin by partitioning it per operation:
\begin{lem}[$M_{2}$ Segment Work]
\label{lem:M2-segment-work} We can divide the \textbf{segment work}
(work done on segments) in $M_{2}$ among the operations in the following
natural way --- each segment $S[k]$ does $\Theta\left(2^{k}\right)$
work per operation it processes.\end{lem}
\begin{proof}
If $k<m$, then the proof is the same as for \nameref{lem:M1-segment-work}
(\ref{lem:M1-segment-work}). So consider only $k\ge m$. The front-locking
takes $O(k)\wi O\left(2^{k}\right)$ work, and both accessing the
filter and inserting into the buffer of $S[k+1]$ take $O(\log p)\wi O\left(2^{k}\right)$
work per operation. Accessing the parallel 2-3 trees takes $\Theta\left(2^{k}\right)$
work per item by \ref{cor:M2-segment-access-bound}. Thus searching
for the accessed items in $S[k]$ takes $O\left(2^{k}\right)$ work
per operation, and frontward transfers (in \step\ref{enu:underfull-transfer})
can be paid for by the successful deletions.

All that remains is to pay for the rearward transfers (in \step\ref{enu:overfull-transfer}).
Define the charge for a rearward transfer from $S[k-1]$ to $S[k]$
to be $2^{k-1}$. Let $C$ be the minimum total charge for rearward
transfers needed to make every segment not exceed capacity. Then each
search/update/insertion that finishes in $S[k]$ increases $C$ by
at most the total charge for a cascade of rearward transfers that
returns the shifted/inserted item to its original segment, which is
less than $2^{k}$. And each rearward transfer done by $S[k]$ is
necessary to make $S[k-1]$ not exceed capacity, so it decreases $C$
by at least $2^{k}$. Therefore the searches/updates/insertions can
pay for the rearward transfers.
\end{proof}
We are now ready to bound the core work done by $M_{2}$:
\begin{lem}[$M_{2}$ Core Work]
\label{lem:M2-core-work} $M_{2}$ takes $O(W_{L})$ segment work
on segment-bound operations, for some linearization $L$ of $D$.\end{lem}
\begin{proof}
As we did for $M_{1}$, we pretend that a deleted item is marked rather
than removed, and marked items are transferred to the next segment
when it runs, and are removed only at the last segment.

Then we use the \nameref{lem:work-set-cost} (\ref{lem:work-set-cost})
on the list $R$ of the items in $M_{2}$ (including marked items)
in order of segment followed by recency within the segment, where
$R$ is updated when each segment $S[k]$ finishes processing each
batch as follows:
\begin{itemize}
\item If $k<m$, then update the sublist of $R$ for $S[0..k]$.
\item If $k\ge m$, then update the sublist of $R$ for $S[k-1]$ and $S[k]$,
and if $S[k+1]$ was just created to hold newly inserted items then
update the sublist of $R$ for $S[k+1]$ as well.
\end{itemize}
To apply the lemma, we construct a sequence $G$ of list operations
on $R$ that simulate the updates to $R$ in the actual execution
$A$ of $M_{2}$ as follows:
\begin{itemize}
\item Shift successfully searched/updated items in $A$: Search for them
in reverse order (from back to front in $R$).
\item Shift marked (to-be-deleted) items in $A$: Demote them.
\item Insert items in $A$: Insert in the desired positions.
\item Remove marked items in $A$: Delete them.
\end{itemize}
Each segment $S[k]$ takes $O\left(2^{k}\right)$ work per operation
that reaches it, by \nameref{lem:M2-segment-work} (\ref{lem:M2-segment-work}).
Note that each searched item in $S[k]$ where $k\ge m$ has rank in
$R$ at least $2^{2^{k-1}}$ before the shift and at most $2\cdot2^{2^{m\smash{'}-1}}$
after the shift where $m'=\min(m,k-1)<k$, since at that point $S[0..m'-1]$
does not exceed capacity by \nameref{lem:M2-balance} (\ref{lem:M2-balance}).
So just as for $M_{1}$, the conditions of \ref{lem:work-set-cost}
are satisfied for $G$ on $R$, and hence $M_{2}$ takes $O(W_{G})$
segment work on segment-bound operations. Now let $L$ be the same
as $G$ but with each operation expanded to the original sequence
of operations that were finished together with that one (according
to the list in the filter entry). Clearly $W_{G}\le W_{L}$, so we
are done.
\end{proof}
We now have the needed lemmas to bound the effective work of $M_{2}$.
\begin{thm}[$M_{2}$ Effective Work]
\label{thm:M2-work} $M_{2}$ takes $O(W_{L}+e_{L}\cdot\log p)$
effective work for some linearization $L$ of $D$.\end{thm}
\begin{proof}
We shall follow the same techniques as in the proof of \nameref{thm:M1-work}.
We can ignore the work done to transform the input batches from the
parallel buffer into cut batches of size $p^{2}$, since it takes
$O(1)$ work per operation. Sorting the cut batches takes $O(W')$
work where $W'=W_{L}+e_{L}\cdot\log p$ for some linearization $L$
by \nameref{lem:batch-sort-cost} (\ref{lem:batch-sort-cost}) with
the choice of $\e=\frac{1}{2}$.

We divide the segment work done among the operations as per \ref{lem:M2-segment-work}.
Then the segment-bound operations take $O(W_{L'})$ total segment
work for some linearization $L'$ of $D$, by \nameref{lem:M2-core-work}
(\ref{lem:M2-core-work}). We can ignore the work done by $S[m-1]$
and the filter, since that is $O(\log p)$ work per operation, each
of which had already taken $\Omega(\log p)$ segment work in $S[0..m-2]$.
Thus we just have to show that the filter-bound operations take $O(W_{L'})$
total work in $S[0..m-2]$.

We classify each time step as a \textbf{high-busy step} iff $Q_{1}$
has at least $\frac{1}{2}p$ ready nodes, and as a \textbf{high-idle
step} otherwise. On each high-busy step, $\frac{1}{2}p$ ready nodes
generated for the final slab will be executed (by the weak-priority
scheduler). Since the final slab takes $O(W_{L'})$ work in total
by the preceding analysis, there can be only $O\left(\frac{W_{L'}}{p}\right)$
high-busy steps, which hence take $O(W_{L'})$ work.

Now consider each first slab run ($M_{2}$ interface \step\ref{enu:first-slab-run})
as comprising an $S[0..m-2]$ run and then an $S[m-1]$ run (where
only the latter is neighbour-locked and front-locked). We classify
each $S[0..m-2]$ run as a \textbf{high-busy run }iff at least half
the time steps during the run are high-busy steps, and as a \textbf{high-idle
run} otherwise. During each high-busy run, the total work is $O(1)$
times the work during high-busy steps, since every high-busy step
does $\Theta(p)$ work. Thus we can ignore the work done during high-busy
runs, and it remains to show that high-idle runs do $O(W_{L'})$ work
on filter-bound operations.

Note that every filter-bound operation is trapped in the filter due
to some operation that finishes in the final slab. So take any operation
$X$ on item $x$ that finishes in segment $S[k]$ in the final slab,
and consider the time interval $Z$ when $X$ is in the final slab
(after being put in the buffer for $S[m]$). Let $\Gamma$ be the
remaining computation for the processing of $X$, which initially
has delay $O\left(2^{k}\right)$ by \nameref{lem:M2-final-slab-bound}
(\ref{lem:M2-final-slab-bound}), and note the following:
\begin{itemize}
\item At any high-idle step, all ready $Q_{1}$-nodes are executed.
\item During any $S[0..m-2]$ run, the interface is not holding any neighbour-lock
or front-lock.
\end{itemize}
These imply that during $\Gamma$, at any high-idle step during an
$S[0..m-2]$ run, the delay of $\Delta(\alpha)$ is reduced for each
acquire-stall node $\alpha$ in $\Gamma$ (by structural induction),
and hence the delay of $\Gamma$ is reduced. Thus there are $O\left(2^{k}\right)$
such steps during $Z$. Finally observe that each $S[0..m-2]$ run
takes $\Omega(\log p)$ high-idle steps, and hence at most $O\left(1+\frac{2^{k}}{\log p}\right)$
$S[0..m-2]$ runs overlap $Z$, and they do at most $O\left(\log p+2^{k}\right)\wi O\left(2^{k}\right)$
work on $x$, which is $O(1)$ times the work done on $X$. Summing
over all operations $X$ that finish in the final slab, at most $O(W_{L'})$
total work is done during $S[0..m-2]$ runs on filter-bound operations.
\end{proof}
At last, we tackle the effective span bound for $M_{2}$.
\begin{defn}[$M_{2}$ Time Linearization]
 We say that a linearization $L$ of $D$ is a \textbf{time linearization}
if the following hold:
\begin{enumerate}
\item $L$ orders operations by when they finish (for those finished by
the final slab this is defined to be when \step\ref{enu:finish-ops}
is run).
\item For operations that finish at the same time in the same segment, $L$
puts searches/updates/insertions in reverse order of how the items
are inserted into $S[m']$ (in \step\ref{enu:finish-ops}).
\end{enumerate}
\end{defn}
\begin{lem}[$M_{2}$ Rank Invariant]
\label{lem:M2-rank-inv} Every item $x$ in the final slab is within
the first $r$ items of the final slab, where $r$ is the number of
distinct items searched/updated/inserted since the last (combined)
operation that shifted/inserted $x$ (in \step\ref{enu:finish-ops}),
according to any time linearization.\end{lem}
\begin{proof}
The shifts/inserts can be simulated by individual shifts/inserts of
items in the same order as the time linearization, such that every
item $x$ that is shifted/inserted is placed at the front of $S[m']$,
at which point the invariant is preserved for $x$. Between consecutive
accesses to $x$, every item searched/updated/inserted can move $x$
rearward in $M_{2}$ by at most $1$ position.\end{proof}
\begin{thm}[$M_{2}$ Effective Span]
\label{thm:M2-span} $M_{2}$ takes $O\left(\frac{W_{L}}{p}+d\cdot(\log p)^{2}+s_{L}\right)$
effective span for some linearization $L$ of $D$.\end{thm}
\begin{proof}
Similar to the proof of \nameref{thm:M1-span}, we shall find the
time taken to execute the actual execution DAG for $P$ using $M_{2}$
on an unlimited number of processors when each $M_{2}$-node takes
unit time while every other node takes zero time. We classify each
time step as a \textbf{filter-full} step iff the filter has size at
least $p$, and as a \textbf{filter-empty} step otherwise. And we
shall separately count them.

In both cases we shall utilize the \nameref{lem:M2-final-slab-bound}
(\ref{lem:M2-final-slab-bound}).

\br

\uline{Filter-full steps}

Each operation $o$ that finishes in a final slab segment $S[k]$
takes $\Theta\left(2^{k}\right)$ work in the final slab as per \nameref{lem:M2-segment-work}
(\ref{lem:M2-segment-work}), but stays in the final slab for only
$O\left(2^{k}\right)$ steps by \ref{lem:M2-final-slab-bound}. Let
$K$ be the collection of all pairs $(i,o)$ such that an operation
$o$ is in a final slab segment at step $i$. Then $\#(K)\in O(W_{L})$
for some linearization $L$ of $D$ by \nameref{lem:M2-core-work}
(\ref{lem:M2-core-work}), and hence there are $O\left(\frac{W_{L}}{p}\right)$
filter-full steps since at least $p$ operations are in the final
slab at each filter-full step.

\br

\uline{Filter-empty steps}

Put a counter at each $M_{2}$-call in the program DAG $D$, initialized
to zero, and on each filter-empty step increment the counter at every
pending $M_{2}$-call. Then the number of filter-empty steps is at
most the final counter-weighted span of $D$, which we shall now bound.

Take any path $C$ in $D$. Let $L'$ be the time linearization of
$D$ that puts each $M_{2}$-call along $C$ before all the other
operations that finish together with it (this is permissible because
all operations that finish together are independent in $D$). Consider
each $M_{2}$-call $X$ on $C$ that accesses item $x$. We will trace
the `journey' of $X$ from the parallel buffer in an uncut batch
$U$ of size $u$ to a cut batch $B$ to the end of $M_{2}$. Let
$r$ be the access rank of $X$ according to $L'$. For convenience
we first observe the following:
\begin{itemize}
\item Each uncut batch of size $u$ takes $O(\log p+\log u)$ steps to be
flushed from the parallel buffer, and $O\left(\log u+\frac{u}{p^{2}}\right)$
steps to be cut and added/appended to the bunches in the feed buffer.
In total this is $O\left(\log p+\log u+\frac{u}{p^{2}}\right)\wi O\left(\log p+\frac{u}{p}\right)$
steps.
\item Each cut batch has size at most $p^{2}$, and hence takes $O(\log p)$
steps to be converted from the bunch in the feed buffer, $O\left((\log p)^{2}\right)$
steps to be sorted and combined, and $O\left(\log p+2^{k}\right)$
steps in segment $S[k]$ for each $k\in[0..m-1]$. ($S[m-1]$ takes
$O(\log p)$ delay by \nameref{lem:M2-front-access-bound} (\ref{lem:M2-front-access-bound}).)
Thus each cut batch takes $O\left((\log p)^{2}\right)$ steps in the
first slab, since $m\in\log O(\log p)$. After that, filtering the
remaining batch takes $O(\log p)$ steps. In total this is $O\left((\log p)^{2}\right)$
steps.
\end{itemize}
At the start, $X$ waits in the parallel buffer for the previous uncut
batch of size $u'$ to be processed, taking $O\left(\log p+\frac{u'}{p}\right)$
steps. Next, the $M_{2}$ interface processes the current cut batch
$B'$, taking $O\left((\log p)^{2}\right)$ steps. Then some filter-full
steps elapse before $M_{2}$ is again ready, at which point $U$ is
processed, taking $O\left(\log p+\frac{u}{p}\right)$ steps. After
that, $X$ waits for some $i$ intervening cut batches, each of which
has size $b^{*}=p^{2}$ and hence takes $O\left((\log p)^{2}\right)\wi O\left(\frac{b^{*}}{p}\right)$
filter-empty steps each. (More filter-full steps may elapse in-between
the cut batches.)

Then, $X$ is passed through the segments of $M_{2}$ in some combined
operation $Y$. If $Y$ finishes in the first slab, then it takes
$O\left((\log p)^{2}\right)$ steps to finish. Otherwise, let $Z$
be $Y$ if $Y$ is segment-bound, and otherwise let $Z$ be the prior
operation on the same item $x$ that was in the final slab when $Y$
gets trapped in the filter. Then $Z$ finishes in some final slab
segment $S[k]$. Just before that $S[k]$ run, $x$ must be within
the first $r$ items of $S[m..l]$ by \nameref{lem:M2-rank-inv} (\ref{lem:M2-rank-inv}),
and $S[0..k-1]$ is at most $2p^{2}$ below capacity by \nameref{lem:M2-balance}
(\ref{lem:M2-balance}), and so $\sum_{i=0}^{m-1}2^{2^{i}}+r\ge\sum_{i=0}^{k-1}2^{2^{i}}-2p^{2}$,
which gives $2^{2^{m}}+r\ge2^{2^{k-1}}-2p^{2}$ and hence $2^{k-1}\le\log\left(r+2^{2^{m}}+2p^{2}\right)\in O(\log p+\log r)$.
Thus $Z$ takes $O(\log p+\log r)$ steps in the final slab, and the
filter entry for $x$ has $O(\log p+\log r)$ combined operations,
and so $Y$ will finish within $O(\log p+\log r)$ steps after $Z$.

Hence $Y$ takes $O\left((\log p)^{2}+\log r\right)$ steps to be
passed through the segments of $M_{2}$, after which the result for
$X$ is returned within $O(\log p)$ steps (since $Y$ comprises at
most $p^{2}$ operations). Thus $X$ takes $O\left(\log p+\frac{u}{p}+\frac{u'}{p}+(\log p)^{2}+\frac{i\cdot b^{*}}{p}+\log r\right)$
steps in total. Also, no two $M_{2}$-calls on the path $C$ can wait
for the same intervening batch. Thus over all counters at $M_{2}$-calls
on $C$, each of $u,u',i\cdot b^{*}$ will sum up to at most the total
number $N$ of $M_{2}$-calls, and hence the counter-weighted span
of $C$ is $O\left(\frac{N}{p}+d\cdot(\log p)^{2}+s_{L'}\right)$.

Therefore the final counter-weighted span of $D$ is $O\left(\frac{N}{p}+d\cdot(\log p)^{2}+s_{L^{*}}\right)$
where $L^{*}$ is the linearization of $D$ that maximizes $s_{L^{*}}$,
as desired.
\end{proof}
This concludes our analysis of $M_{2}$.

\section{Practical Schedulers}

\label{sec:work-steal}

The bounds on the effective work and span in \ref{sec:PWM1} and \ref{sec:PWM2}
apply if we use a greedy scheduler for $M_{1}$ and a weak-priority
scheduler (\ref{sub:weak-priority-sched}) for $M_{2}$. In practice,
we do not have such schedulers. But a work-stealing scheduler~\footnote{To ensure that each processor can access its own deque in $O(1)$
time, the deque is guarded by a dedicated lock (\ref{def:dedi-lock})
with $2$ keys, $1$ for the processor and $1$ for the external interface
through which other processors access the deque. The external interface
is guarded by another dedicated lock with $1$ key for each processor.} for $M_{1}$ gives the desired time bound (\ref{thm:M1-perf}) on
average, as essentially shown in \cite{blumofe1999worksteal,arora2001worksteal}.~\footnote{The results in \cite{blumofe1999worksteal} are for strict computations
in an atomic message passing model, but the proof of the time bound
carries over to the QRMW parallel pointer machine model in the same
manner as done in \cite{arora2001worksteal}.}

As for $M_{2}$, dedicating $\frac{1}{2}p$ processors to a greedy
scheduler for $Q_{1}$-nodes and the other $\frac{1}{2}p$ processors
to another greedy scheduler for $Q_{2}$-nodes gives a weak-priority
scheduler as required for $M_{2}$. Replacing each greedy scheduler
by a work-stealing scheduler should yield the desired time bound (\ref{thm:M2-perf})
on average.

\section{Conclusions}

This paper presents two parallel working-set data structures, both
nearly achieving the working-set bound in their effective work, and
the faster version having a lower overhead in its effective span by
using careful pipelining. Pipelining techniques to reduce the span
of data structure operations have been explored before~\cite{BlellochRe97}.
Our results indicate that implicit batching, especially combined with
pipelining, has promise in the design of parallel data structures.
As a future research direction, it would be interesting to see if
these ideas apply to other data structures, in particular, self-adjusting
data structures such as splay trees that provide good amortized performance
(rather than worst-case performance) and/or randomized data structures.

\renewcommand{\thesection}{\hspace{-1em}}

\section{Appendix}

\renewcommand{\thesection}{A}

Here we spell out the data structures, locking mechanisms and supporting
theorems that we have used in our paper.

\subsection{Parallel Buffer}

\label{sub:par-buffer}

We can use any parallel buffer implementation that takes $O(p+b)$
work and $O(\log p+\log b)$ span per batch of size $b$, and is such
that (regardless of the scheduler) any operation that arrives will
be included in the batch that is being flushed or in the next batch,
and it always has at most $\frac{1}{2}p+q$ ready nodes (active threads)
where $q$ is the number of operations that are currently buffered
or being flushed. Then the parallel buffer overhead is given by the
following theorem.
\begin{thm}[Parallel Buffer Cost]
\label{rem:par-buff-cost} Take any program $P$ using a batched
data structure $M$ on $p$ processors using any greedy scheduler.
Let $E$ be the actual execution DAG of $P$ using $M$. The \textbf{buffer's
effective cost} (see \ref{def:effective}) is defined as $\frac{t_{1}}{p}+t_{\infty}$,
where $t_{1}$ is the total number of buffer nodes in $E$, and $t_{\infty}$
is the maximum number of buffer nodes on any path in $E$. Then the
buffer's effective cost is $O\left(\frac{T_{1}+w_{M}}{p}+d\cdot\log p\right)$,
where $w_{M}$ is the effective work taken by $M$ and $d$ is the
maximum number of $M$-calls on any path in the program DAG $D$ for
$P$.\end{thm}
\begin{proof}
Let $N$ be the total number of operations on $M$. Consider each
batch $B$ of $b$ operations on $M$. Let $t_{B}$ be span taken
by the buffer on $B$. If $b\le p^{2}$, then $t_{B}\in O(\log p)$.
If $b>p^{2}$, then $t_{B}\in O(\log b)\wi O\left(\frac{b}{p}\right)$.
Thus $t_{B}\in O\left(\frac{b}{p}+\log p\right)$ and hence $t_{\infty}\in O\left(\frac{N}{p}+d\cdot\log p\right)$.

Now consider the execution of $E$. At each time step, the buffer
is processing at most two consecutive batches, so we shall analyze
the buffer work done during the time interval for each pair of consecutive
batches $B$ and $B'$, where $B$ has $b$ operations and $B'$ has
$b'$ operations.

If $b+b'\ge\frac{5}{6}p$, then the buffer work done on $B$ and $B'$
is $O(b+b')$.

If $b+b'<\frac{1}{6}p$, then at least one of the following hold at
each time step in this interval:
\begin{itemize}
\item At least $\frac{1}{6}p$ ready $P$-nodes in $E$ are executed. These
steps take at most $O(T_{1})$ work over all intervals.
\item At least $\frac{1}{6}p$ ready $M$-nodes in $E$ are executed. These
steps take at most $O(w_{M})$ work over all intervals.
\item At least $\frac{2}{3}p$ ready buffer nodes in $E$ are executed,
which is impossible since $b+b'+\frac{1}{2}p<\frac{2}{3}p$.
\item Less than $p$ ready nodes in $E$ are executed. All ready buffer
nodes in $E$ are executed (by the greedy scheduling), so over all
intervals there are $O(t_{\infty})$ such steps, taking $O(p\cdot t_{\infty})$
work.
\end{itemize}
Therefore $\frac{t_{1}}{p}\in O(\frac{T_{1}+w_{M}}{p}+t_{\infty})$,
and hence the buffer's effective cost is $\frac{t_{1}}{p}+t_{\infty}\in O\left(\frac{T_{1}+w_{M}}{p}+d\cdot\log(p)\right)$
since $N\le T_{1}$.
\end{proof}
The parallel buffer for each data structure can be implemented in
the QRMW pointer machine model using a static BBT (balanced binary
tree), with a sub-buffer (and its length) at each leaf node, with
one sub-buffer for each processor, and a flag at each internal node.
On every call to the data structure, the processor terminates the
current thread and puts its continuation together with the call parameters
into the sub-buffer for that processor. Then the processor walks up
the BBT from leaf to root, test-and-setting each flag along the way,
terminating if it was already set. On reaching the root, the processor
activates the data structure's interface, which when ready will then
flush the buffer, via a parallel recursion on the BBT to atomically
swap out all the sub-buffers and then combine their contents in parallel,
which it then passes to the data structure itself as a batch. The
recursion simultaneously constructs a new static BBT for the new sub-buffers
that have been swapped in, which the processors use for subsequent
calls to that data structure.

\begin{figure}[h]
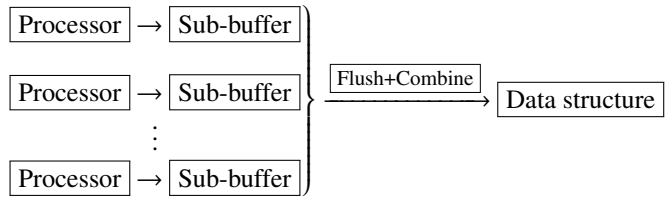

\noindent \begin{centering}
$\left.\begin{matrix}\tbox{Processor}\to\tbox{Sub-buffer}\\
\\
\tbox{Processor}\to\tbox{Sub-buffer}\\
\vdots\\
\tbox{Processor}\to\tbox{Sub-buffer}
\end{matrix}\right\} \xrightarrow{\stbox{Flush+Combine}}\tbox{Data\ structure}$
\par\end{centering}

\caption{Data flow in the parallel buffer}
\end{figure}

Test-and-setting each flag in the BBT takes constant time, because
at most two processors ever access it. The parallel buffer takes $O(p)$
work and $O(\log(p))$ span to initialize. Each data structure call
takes $O(p)$ work and $O(\log(p))$ span for a processor to reach
the root, because the flags ensure that only $O(1)$ work is done
per node in traversing the BBT, and flushing the buffer on a batch
of size $b$ takes $O(p+b)$ work and $O(\log(p)+\log(b))$ span.
Thus the total cost for the batch is $O(p+b)$ work and $O(\log(p)+\log(b))$
span.

It remains to implement the sub-buffer to support $O(1)$ worst-case
time insert but $O(\log(b))$ flushing where $b$ is the current size.
This can be done by filling in a complete binary tree in each level
from left to right, assisted by maintaining a linked list through
each level and a pointer to the leftmost leaf node, which we can construct
from the linked list for the parent level as we fill in the current
level. Flushing is trivial.

\subsection{Batched Parallel 2-3 Tree}

\label{sub:par-2-3-tree}

The batched parallel 2-3 tree that we use in our parallel working-set
maps can be implemented by adapting the parallel 2-3 dictionary described
in \cite{paul1983paradict} to the QRMW pointer machine model. A batched
parallel 2-3 tree $M$ supports the following batched operations (where
$n$ is the number of items in $M$ before the batch):
\begin{enumerate}
\item \textbf{Normal batch operation:} Given any item-sorted input batch
$B$ of $b$ operations on distinct items, $M$ performs all the operations
in $B$, and returns an output batch $B'$ containing the results
for the operations stored in the same order as in $B$, all within
$\Theta(b\cdot\log n)$ work and $O(\log b+\log n)$ span.

\begin{itemize}
\item The \textbf{result} for each operation is a direct pointer\textbf{
}to the item in $M$ (or $null$ if deleted).
\item A \textbf{direct pointer} is an object from which one can access the
item itself (and read/modify any attached value). 
\end{itemize}

\item \textbf{Reverse-indexing operation:} Given any (unsorted) input batch
$B$ of $b$ direct pointers to distinct items in $M$, $M$ returns
an output batch $B'$ that is an item-sorted batch of the same items
pointed to by the direct pointers in $B$, all within $\Theta(b\cdot\log n)$
work and $O(\log b+\log n)$ span. 
\end{enumerate}

\subsection{Sorting Theorems}

In the comparison model the items in the search problem can come from
any arbitrary set $S$ that is linearly ordered by a given comparison
function. We shall also assume that $S$ has at least two items. As
is standard, let $S^{n}$ be the class of all length-$n$ sequences
from $S$. Search structures can often be adapted to implement a sorting
algorithm~\footnote{A sorting algorithm on a class $C$ of sequences is a procedure that
given any input from $C$ will output a list of pointers that corresponds
to the input in sorted order.}, in which case any lower bound on complexity of sorting sequences
from $S^{n}$ typically implies a lower bound on the costs of the
search structure with $n$ items. For the proofs of \nameref{thm:M1-work}
and \nameref{thm:M2-work} we need a crucial lemma that the entropy
bound is a lower bound for (comparison-based) sorting up to a constant
factor, even in the average case, as precisely stated below.
\begin{defn}[Normalized Frequencies]
 We say that $q_{1..u}$ are \textbf{normalized frequencies} if $q_{1..u}\in\rr^{+}$
and $\sum_{i=1}^{u}q_{i}=1$. Let $\f{Seq}(n,q)$ be the class of
sequences in $S^{n}$ with normalized item frequencies $q$.\end{defn}
\begin{thm}[Sorting Entropy Bound]
\label{thm:sort-entr-bound} Take any normalized frequencies $q_{1..u}$,
and let $C=\f{Seq}(n,q)$. Then any sorting algorithm $A$ on $S^{n}$
requires at least $\max\left(\frac{n-1}{5}\cdot H,n-1\right)$ comparisons
on average for input sequences from $C$, where $H=\sum_{i=1}^{u}\left(q_{i}\cdot\ln\left(\frac{1}{q_{i}}\right)\right)$
is the entropy of $C$ per element. Asymptotically, $A$ requires
$\Omega(n\cdot H+n)$ comparisons on average for input sequences from
$C$.~\footnote{In fact no sorting algorithm takes $O(n\cdot H)$ steps on every sequence
from $C$, since we require the implicit constant to not depend on
the access distribution.}\end{thm}
\begin{proof}
Take any correct sorting algorithm $A$ on $S^{n}$. If $u=1$ or
$n=1$ then $\frac{n-1}{5}\cdot H=0$ and so the claim holds. Therefore
we can assume that $u>1$ and $n>1$. If $H\le5$ then $\frac{n-1}{5}\cdot H\le n-1$
and so the claim holds, since $A$ needs at least $n-1$ comparisons,
otherwise if every comparison returns equality then the graph with
comparisons as edges has at least $2$ connected components and it
cannot be determined whether their items are equal or not. Therefore
we can assume that $H>5$.

Note that for every $x\in\zz^{+}$ we have $\int_{1}^{x}\ln(x)\le\ln(x!)\le\frac{1}{2}\ln(1)+\int_{1}^{x}\ln(x)+\frac{1}{2}\ln(x)$
since $\ln$ is concave, and hence $x\cdot\ln(x)-x+1\le\ln(x!)\le x\cdot\ln(x)-x+1+\frac{1}{2}\ln(x)$.
Thus $\ln(\#(C))=\ln\left(\frac{n!}{\prod_{i=1}^{u}(n\cdot q_{i})!}\right)$
$\ge n\cdot\ln(n)-n-\sum_{i=1}^{u}\left(n\cdot q_{i}\cdot\ln(n\cdot q_{i})-n\cdot q_{i}+1+\frac{1}{2}\ln(n\cdot q_{i})\right)$
$=n\cdot H-u-\frac{1}{2}\sum_{i=1}^{u}\ln(n\cdot q_{i})$ $\ge n\cdot H-n-\frac{u}{2}\cdot\ln(n)$.
Also we have $\frac{n}{2}\cdot H\ge\frac{n}{2}\cdot(u-1)\cdot\left(\frac{1}{n}\cdot\ln(n)\right)=\frac{u-1}{2}\cdot\ln(n)$
by smoothing, and hence $\ln(\#(C))\ge\frac{n}{2}\cdot H-n-\frac{1}{2}\cdot\ln(n)$.
And since $\frac{\ln(n)}{n}\le\frac{1}{e}$, we get $\ln(\#(C))\ge n\cdot\left(\frac{1}{2}H-1-\frac{1}{2e}\right)>n\cdot\frac{\ln(3)}{5}\cdot H$,
and hence $\log_{3}(\#(C))\ge\frac{n}{5}\cdot H>\frac{n-1}{5}\cdot H$.

Let $T$ be the ternary tree corresponding to all possible executions
of $A$ on sequences from $C$, where each node corresponds to a comparison
whose outcome determines the subtree that the execution will proceed
to, and each leaf node corresponds to a terminal state. For each node
$v$ in $T$, let $N(v)$ be its child nodes, and let $c(v)$ be the
number of leaf nodes reachable from $v$, and let $f(v)$ be the number
of comparisons needed from that point on average. Then $f(v)=0$ if
$v$ is a leaf node, and $f(v)=1+\frac{1}{c(v)}\cdot\sum_{w\in N(v)}(c(w)\cdot f(w))$
otherwise. If $f(w)\ge\log_{3}(c(w))$ for every $w\in N(v)$, then
$f(v)\ge1+\log_{3}\left(\frac{1}{3}\cdot\sum_{w\in N(v)}c(w)\right)$
by Jensen's inequality, and hence $f(v)\ge\log_{3}(c(v))$. Thus by
structural induction, $f(r)\ge\log_{3}(c(r))$ where $r$ is the root
node of $T$, and hence $A$ needs at least $\log_{3}(\#(C))$ comparisons
on average over all input sequences from $C$.

Therefore $A$ needs $\max\left(\frac{n-1}{5}\cdot H,n-1\right)\in\Omega(n\cdot H+n)$
comparisons on average over all input sequences from $C$.
\end{proof}
It also turns out that there is a sequential sorting algorithm $ESort$
on $S^{n}$ that achieves the entropy bound (up to a linear discrepancy),
as stated precisely in the subsequent definition and theorem.
\begin{defn}[Sequential Entropy Sort]
\label{def:seq-esort} Let $ESort$ be the algorithm that does the
following on an input sequence $I$:
\begin{block}
\item Let $D$ be a dictionary that stores items tagged with a list of pointers,
implemented using Iacono's working-set structure. Iterate through
each item in $I$, and insert it into $D$ (with an empty list) if
it is not already there, and then append its location in $I$ to its
list in $D$. For each segment in $D$, construct the sorted list
of the (tagged) items in that segment. Merge these sorted lists in
order of increasing capacity to obtain the sorted list $L$ of the
items in $I$, in which the tag of each item identifies all its duplicates
in $I$. Expand each item in $L$ to its tag, and output the resulting
list of pointers.
\end{block}
\end{defn}
\begin{thm}[$ESort$ Performance]
\label{thm:seq-esort-perf} Take any normalized frequencies $q_{1..u}$,
and let $C=\f{Seq}(n,q)$. Then $ESort$ takes $\Theta(W)\wi O(n\cdot H+n)$
steps on every sequence $I\in C$, where $W$ is the insert working-set
bound (\ref{def:ins-work-set}) for $I$ and $H=\sum_{i=1}^{u}\left(q_{i}\cdot\ln\left(\frac{1}{q_{i}}\right)\right)$.\end{thm}
\begin{proof}
Note that $D$ satisfies the working-set property, and hence $ESort$
clearly takes $\Theta(W+n)=\Theta(W)$ steps on $I$. Let $c_{i}=n\cdot q_{i}$
be the number of occurrences of item $i$ in the sequence, and $r_{i}(j)$
be the access rank of the $j$-th occurrence of item $i$. Then the
first occurrence of each item takes $O(\log u+1)$ steps since its
access rank is at most $u$. By smoothing we get $n\cdot H\ge n\cdot(u-1)\cdot\left(\frac{1}{n}\cdot\ln n\right)\ge(u-1)\cdot\ln u$,
and hence all the first accesses to the items takes $O(u\cdot\log u+u)\wi O(n\cdot H+\ln u+u)\wi O(n\cdot H+n)$
steps. Also, for each item $i$ that occurs more than once, the subsequent
accesses to item $i$ takes $O\left(\sum_{j=2}^{c_{i}}\left(\log r_{i}(j)+1\right)\right)<O\left(\sum_{j=2}^{c_{i}}\log r_{i}(j)+c_{i}\right)$
steps, and Jensen's inequality gives $\sum_{j=2}^{c_{i}}\log r_{i}(j)\le(c_{i}-1)\cdot\log\left(\frac{1}{c_{i}-1}\sum_{j=2}^{c_{i}}r_{i}(j)\right)\le c_{i}\cdot\log\left(\frac{2}{c_{i}}\cdot n\right)=c_{i}\cdot\log\frac{1}{q_{i}}+c_{i}$,
and hence all the subsequent accesses to the items takes $O\left(\sum_{i=1}^{u}\left(n\cdot q_{i}\cdot\log\frac{1}{q_{i}}+c_{i}\right)\right)=O(n\cdot H+n)$
steps. Thus $W\in O(n\cdot H+n)$. Also, merging the segments into
$L$ takes $O(u)$ steps because each segment is at least twice the
size of the preceding one (except possibly the last). Finally, expanding
each item in $L$ to its tag takes $O(n)$ steps. Therefore $ESort$
takes $\Theta(W)\wi O(n\cdot H+n)$ steps.\end{proof}
\begin{thm}[Worse-case Working-set Bound]
\label{thm:work-set-worst} Take any normalized frequencies $q_{1..u}$,
and let $C=\f{Seq}(n,q)$. Then the working-set bound for inserting
some sequence $I\in C$ is $\Omega(n\cdot H+n)$ where $H=\sum_{i=1}^{u}\left(q_{i}\cdot\ln\frac{1}{q_{i}}\right)$.\end{thm}
\begin{proof}
By the \nameref{thm:sort-entr-bound} let $I\in C$ such that $ESort(I)$
takes $\Omega(n\cdot H+n)$ comparisons, and let $W$ be the working-set
bound for inserting $I$. Since $ESort(I)$ takes $\Theta(W)$ steps
(\ref{thm:seq-esort-perf}), we have $W\in\Omega(n\cdot H+n)$.
\end{proof}
Finally we give a parallel sorting algorithm $PESort$ on $S^{n}$
that achieves the entropy bound for work but yet takes only $O\left((\log n)^{2}\right)$
span, which we need in our parallel working-set map. Note that these
algorithms work in the QRMW pointer machine model, in which input
and output lists are stored in leaf-based BBTs (balanced binary trees
with all items at the leaves).
\begin{defn}[Parallel Entropy Sort]
\label{def:par-esort} Let $PESort$ be the following parallel variant
of Quicksort:
\begin{block}
\item Use the \nameref{lem:par-pivot-alg} to pick a pivot from the two
middle quartiles of the input list. Then partition the list around
the pivot into a lower part and a middle part (equal to the pivot)
and an upper part (parallelized via the standard prefix-sum technique).
Then sort the lower and upper parts recursively. Finally concatenate
the three parts.
\end{block}
\end{defn}
\begin{thm}[$PESort$ Performance]
\label{thm:par-esort-perf} Take any normalized frequencies $q_{1..u}$,
and let $C=\f{Seq}(n,q)$. Then $PESort$ sorts every sequence from
$C$, taking $O(n\cdot H+n)$ work and $O\left((\log n)^{2}\right)$
span, where $H=\sum_{i=1}^{u}\left(q_{i}\cdot\ln\frac{1}{q_{i}}\right)$.\end{thm}
\begin{proof}
Consider any item $x$ with frequency $r$ in the input sequence to
$PESort$. At each pivoting stage, $x$ will end up in either the
middle part with size exactly $r$, or in a part with size at most
$\frac{3}{4}$ of the list size, and hence traverses $O\left(\log\frac{k}{r}\right)$
stages in the recursion. Therefore the partitioning steps and recursive
calls take in total $O\left(\sum_{i=1}^{u}\left(n\cdot q_{i}\cdot\log\frac{n}{n\cdot q_{i}}\right)\right)=O(n\cdot H)$
work. The terminal stages take in total $O(n)$ work. Each concatenation
takes $O(1)$ work because leaf-based BBTs with heights differing
by $O(1)$ can be concatenated in $O(1)$ work. Also, each stage takes
$O(\log n)$ span (\ref{lem:par-pivot-alg}), and the depth of the
recursion is $O(\log n)$, so the total span is $O\left((\log n)^{2}\right)$.\end{proof}
\begin{lem}[Parallel Pivot Algorithm]
\label{lem:par-pivot-alg} Let $PPivot$ be the following parallel
algorithm:
\begin{block}
\item Partition the input list of size $k$ into blocks $B[1..c]$ of size
$\log k$ except perhaps the last block. Then for each $i$ in $[1..c]$
in parallel, find the median $m[i]$ of $B[i]$ using the sequential
linear-time median finding algorithm. Then use an optimal parallel
sorting algorithm (such as adapted from \cite{goodrich1996parallelsort,atallah1989cascading})
to sort $m[1..c]$ and output their median. (If there are two medians
choose either one.)
\end{block}
Then $PPivot$ outputs an item that is in the two middle quartiles
of any input list of size $k$, taking $O(k)$ work and $O(\log(k))$
span.\end{lem}
\begin{proof}
Let $x$ be the output item. Then at least half of the blocks have
median at most $x$ and at least half of the blocks have median at
least $x$, and hence $x$ is not in the first or last quartile of
the input list. The partitioning takes $O(k)$ work and $O(\log k)$
span. Finding the median of each block takes $O(\log k)$ work and
span, and so constructing $m[1..c]$ takes $O(k)$ work and $O(\log k)$
span. Finally, $c\in O\left(\frac{k}{\log k}\right)$ and hence sorting
$m[1..c]$ takes $O(c\cdot\log c)=O(k)$ work and $O(\log c)\wi O(\log k)$
span.\end{proof}
\begin{rem*}
A much easier alternative to the \nameref{lem:par-pivot-alg} is to
repeatedly pick the pivot uniformly randomly until it falls into the
middle quartiles of the input, which succeeds in $O(1)$ expected
attempts. $M_{1}$ and $M_{2}$ will still have the same performance
bounds on average (\ref{thm:M1-perf} and \ref{thm:M2-perf}), because
$PESort$ is only used to sort each cut batch, and so will not affect
the subsequent work done on the batch. Hence this should make no difference
when using work-stealing schedulers in practice (\ref{sec:work-steal}).
\end{rem*}

\subsection{Locking Mechanisms\label{sub:locking}}

In this section we give pseudo-code implementations of the various
locking mechanisms used in our data structures, that are supported
in the QRMW pointer machine model.

The \textbf{non-blocking lock} is trivially implemented using the
test-and-set operation as shown in TryLock/Unlock below.
\begin{defn}[Non-blocking Lock]
~
\begin{block}
\item \textbf{TryLock( Bool $x$ ):}

\begin{block}
\item Return $\neg\f{TestAndSet}(x)$.
\end{block}
\item \textbf{Unlock( Bool $x$ ):}

\begin{block}
\item Set $x:=false$.
\end{block}
\end{block}
\end{defn}
We can use non-blocking locks to implement the \textbf{activation}
interface for a process $P$, through which a process can start $P$
if it is not running and condition $C$ is true. Note that any process
that makes $C$ become true must activate $P$. If $P$ has to keep
running as long as $C$ is true, then $P$ must reactivate itself.
\begin{defn}[Activation Interface]
\label{def:local-activation}~($P$ is the process to be guarded
by the interface.)
\begin{block}
\item Private Process $P$.\quad{}// Returns $true$ iff it is to be reactivated.
\item Private Process $C$.\quad{}// Returns $true$ iff $P$ is ready
to run.
\item Private Bool $active:=false$.
\item \textbf{Public Activate():}

\begin{block}
\item If $\f{TryLock}(active)$:

\begin{block}
\item Create Bool $reactivate:=false$.
\item If $C()$, then $reactivate:=P()$.
\item $\f{Unlock}(active)$.
\item If $reactivate$, then Activate().
\end{block}
\end{block}
\end{block}
\end{defn}
If $P$ is only activated by $O(1)$ processes at any time, and $C$
runs in $O(1)$ steps, then each activation call completes within
$O(1)$ steps.

The \textbf{dedicated lock} with keys $[1..k]$, where threads must
use distinct keys to acquire it, can be implemented using the fetch-and-add
operation as shown below.
\begin{defn}[Dedicated Lock]
\label{def:dedi-lock}~($k$ is a positive integer.)
\begin{block}
\item Private Int $count:=0$.
\item Private Int $l:=0$.
\item Private Array $q[1..k]$ initialized with $null$.
\item \textbf{Public Acquire( Int $i$ ):}

\begin{block}
\item If $\f{FetchAndAdd}(count,1)=0$:

\begin{block}
\item Set $l:=i$.
\item Return.
\end{block}
\item Otherwise:

\begin{block}
\item Write pointer to continuation of current thread into $q[i]$.
\item Terminate.
\end{block}
\end{block}
\item \textbf{Public Release():}

\begin{block}
\item If $\f{FetchAndAdd}(count,-1)>1$:

\begin{block}
\item Create Int $j:=l$.
\item Create Pointer $p:=null$.
\item While $p=null$:

\begin{block}
\item Set $j:=j\%k+1$.
\item If $q[j]\ne null$, then swap $p,q[j]$.
\end{block}
\item Set $l:=j$.
\item Fork to resume $p$.
\end{block}
\end{block}
\end{block}
\end{defn}
When any thread $\tau$ attempts to acquire the dedicated lock with
$k$ keys using key $i$, it takes $O(k)$ steps for the fetch-and-add.
If it fails to acquire it at that point, then it writes a pointer
to its continuation into $q[i]$. After that, if $j$ is the key used
by the thread currently holding the lock, then all threads that are
resumed before $\tau$ have keys in cyclic order between $j$ and
$i$.

\section*{Acknowledgements}

We would like to express our gratitude to our families and friends
for their wholehearted support, and to all others who have given us
valuable comments and advice. This research was supported in part
by Singapore MOE AcRF Tier 1 grant T1 251RES1719 and National Science
Foundation grants CCF-1150036, CCF-1733873 and CCF-1725647.

\setlength{\baselineskip}{1em}

\bibliographystyle{plain}
\addcontentsline{toc}{section}{\refname}\bibliography{OPWM}

\end{document}